\documentclass[runningheads,a4paper]{llncs}
\usepackage[T1]{fontenc}
\usepackage{algorithm}
\usepackage[noend]{algpseudocode}
\algrenewcommand\algorithmicrequire{\textbf{Input:}}
\algrenewcommand\algorithmicensure{\textbf{Output:}}
\usepackage{amsmath}
\usepackage{amssymb}
\usepackage{stmaryrd}
\usepackage{mathtools}
\usepackage{thm-restate}
\usepackage{csquotes}
\usepackage{xspace}
\usepackage{xcolor}
\usepackage{siunitx}
\usepackage{eurosym}
\usepackage[bookmarks,breaklinks,hidelinks]{hyperref}
\usepackage{color}

\usepackage{cleveref}
\usepackage{tikz}
\usetikzlibrary{
  arrows,
  backgrounds,
  calc,
  positioning,
}
\tikzset{
  basic settings/.style={
    on grid=true,
    node distance=1.7cm,
  },
  drawandfill/.style={
    draw=#1,fill=#1
  },
  every edge/.style={
    line width=0.5pt,
    draw=black,
    >=stealth',
    bend angle=35,
    auto
  },
  role/.style={
    -stealth',
    shorten >=0.5pt,
    draw=#1
  },
  role/.default=black,
  unnamed individual/.style = {
    circle,
    drawandfill=#1,
    inner sep=1pt,
  },
  unnamed individual/.default = gray!30,
}
\newcommand{\PSpace}{\upshape{\textsc{PSpace}}\xspace}
\newcommand{\ExpTime}{\upshape{\textsc{Exp\-Time}}\xspace}
\newcommand{\NExpTime}{\upshape{\textsc{NExp\-Time}}\xspace}
\newcommand{\Amc}{{\ensuremath{\mathcal{A}}}\xspace}
\newcommand{\Emc}{{\ensuremath{\mathcal{E}}}\xspace}
\newcommand{\Imc}{{\ensuremath{\mathcal{I}}}\xspace}
\newcommand{\Jmc}{{\ensuremath{\mathcal{J}}}\xspace}
\newcommand{\Mmc}{{\ensuremath{\mathcal{M}}}\xspace}
\newcommand{\Omc}{{\ensuremath{\mathcal{O}}}\xspace}

\newcommand{\Tmc}{{\ensuremath{\mathcal{T}}}\xspace}
\newcommand{\Umc}{{\ensuremath{\mathcal{U}}}\xspace}

\newcommand{\Cbb}{{\ensuremath{\mathbb{C}}}\xspace}
\newcommand{\Ibb}{{\ensuremath{\mathbb{I}}}\xspace}
\newcommand{\Nbb}{{\ensuremath{\mathbb{N}}}\xspace}
\newcommand{\Qbb}{{\ensuremath{\mathbb{Q}}}\xspace}
\newcommand{\Tbb}{{\ensuremath{\mathbb{T}}}\xspace}
\newcommand{\Zbb}{{\ensuremath{\mathbb{Z}}}\xspace}
\newcommand{\Bfr}{{\ensuremath{\mathfrak{B}}}\xspace}
\newcommand{\Cfr}{{\ensuremath{\mathfrak{C}}}\xspace}
\newcommand{\efr}{{\ensuremath{\mathfrak{e}}}\xspace}
\newcommand{\NC}{\ensuremath{\mathsf{N_C}}\xspace}
\newcommand{\NR}{\ensuremath{\mathsf{N_R}}\xspace}
\newcommand{\NI}{\ensuremath{\mathsf{N_I}}\xspace}
\newcommand{\NF}{\ensuremath{\mathsf{N_F}}\xspace}
\newcommand{\ex}[1]{\ensuremath{\mathsf{#1}}\xspace}

\newcommand{\ALC}{\ensuremath{\mathcal{ALC}}\xspace}
\newcommand{\ALCQ}{\ensuremath{\mathcal{ALCQ}}\xspace}
\newcommand{\ELD}{\ensuremath{\mathcal{EL}(\cDom)}\xspace}
\newcommand{\ALCD}{\ensuremath{\ALC(\cDom)}\xspace}

\newcommand{\ALCSCC}{\ensuremath{\mathcal{ALCSCC}}\xspace}
\newcommand{\ALCSCCpp}{\ensuremath{\mathcal{ALCSCC}^{++}}\xspace}

\newcommand{\ALCSCCppD}{\ensuremath{\ALCSCCpp(\cDom)}\xspace}
\newcommand{\ALCOSCC}{\ensuremath{\mathcal{ALCOSCC}}\xspace}
\newcommand{\ALCO}{\ensuremath{\mathcal{ALCO}}\xspace}

\newcommand{\ALCOSCCD}{\ensuremath{\mathcal{ALCOSCC}(\stru{D})}\xspace}
\newcommand{\SSCC}{\ensuremath{\mathcal{SSCC}}\xspace}

\DeclareMathOperator{\suc}{\mathsf{succ}}
\DeclareMathOperator{\sat}{\mathsf{sat}}
\newcommand{\dvd}{\mathrel{\mathsf{div}}}
\newcommand{\con}{\ensuremath{\mathfrak{con}}\xspace}
\DeclareMathOperator{\ars}{\mathsf{ars}}
\newcommand{\stru}[1]{\mathfrak{#1}}
\newcommand{\JEPD}{\textsf{JEPD}\xspace}
\newcommand{\JD}{\textsf{JD}\xspace}
\newcommand{\AP}{\textsf{AP}\xspace}
\DeclareMathOperator{\CSP}{CSP}
\newcommand{\cDom}{\ensuremath{\stru{D}}\xspace}
\newcommand{\qDom}{\ensuremath{\stru{Q}}\xspace}
\DeclareMathOperator{\nx}{\mathtt{next}}

\newcommand{\eqterm}[2]{(#1 = \, #2)}
\newcommand{\expadm}{\ExpTime-$\omega$-admissible\xspace}

\newcommand{\nt}{\ensuremath{N_{\Tmc}}\xspace}
\newcommand{\mt}{\ensuremath{M_{\Tmc}}\xspace}
\DeclareMathOperator{\supp}{\mathsf{supp}}
\newcommand{\agt}{\mathfrak{t}}
\newcommand{\ain}{\mathfrak{a}}
\DeclareMathOperator{\agroot}{\mathsf{root}}
\DeclareMathOperator{\wend}{\mathsf{end}}
\newcommand{\cs}{\ensuremath{\mathfrak{C}}\xspace} 
\newcommand{\ocs}{\ensuremath{\mathfrak{B}}\xspace} 
\newcommand{\csi}{\ensuremath{\mathfrak{C}_\Ibb}\xspace}
\renewcommand{\merge}[1]{\ensuremath{\mathop{{\vartriangleleft}_{#1}}}\xspace}
\DeclareMathOperator{\wit}{S_\mathsf{cd}}

\renewcommand{\orcidID}[1]{\,\href{https://orcid.org/#1}{\includegraphics[width=7px]{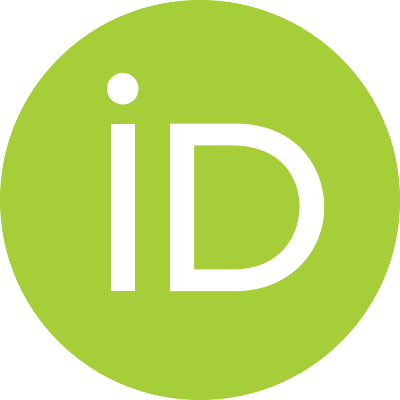}}}
\newcommand{\inAppendix}{in the appendix\xspace}
\begin{document}
\title{Concrete Domains Meet Expressive Cardinality Restrictions in Description Logics \texorpdfstring{\\}{} (Extended Version)}
\titlerunning{Concrete Domains Meet Expressive Cardinality Restrictions in DLs}
\author{Franz Baader\inst{1,2}\orcidID{0000-0002-4049-221X} \and
Stefan Borgwardt\inst{1}\orcidID{0000-0003-0924-8478} \and
Filippo De Bortoli\inst{1,2}\orcidID{0000-0002-8623-6465} \and
Patrick~Koopmann\inst{3}\orcidID{0000-0001-5999-2583}}
\authorrunning{F. Baader et al.}
%
\institute{TU Dresden, Institute of Theoretical Computer Science \\
Dresden, Germany \and
Center for Scalable Data Analytics and Artificial Intelligence (ScaDS.AI) \\
Dresden/Leipzig, Germany \and
Vrije Universiteit Amsterdam, Department of Computer Science \\
Amsterdam, The Netherlands}
\maketitle

\begin{abstract}
Standard Description Logics (DLs) can encode quantitative aspects of an application domain through either \emph{number restrictions}, which constrain the number of individuals that are in a certain relationship with an individual, or \emph{concrete domains}, which can be used to assign concrete values to individuals using so-called features.
These two mechanisms have been extended towards very expressive DLs, for which reasoning nevertheless remains decidable.
Number restrictions have been generalized to more powerful comparisons of sets of role successors in \ALCSCC, while the comparison of feature values of different individuals in \ALCD has been studied in the context of $\omega$-admissible concrete domains \cDom.
In this paper, we combine both formalisms and investigate the complexity of reasoning in the thus obtained DL \ALCOSCCD, which additionally includes the ability to refer to specific individuals by name.
We show that, in spite of its high expressivity, the consistency problem for this DL is \ExpTime-complete, assuming that the constraint satisfaction problem of~\cDom is also decidable in exponential time.
It is thus not higher than the complexity of the basic DL \ALC.
At the same time, we show that many natural extensions to this DL, including a tighter integration of the concrete domain and number restrictions, lead to undecidability.
\end{abstract}
\section{Introduction}
\label{sec:introduction}

Description logics (DLs) \cite{Baader_etal_2003,baader_horrocks_lutz_sattler_2017} are a well-investigated family of logic-based knowledge
representation languages, which can be used to formalize the terminological knowledge of an application domain in a machine-processable way.
For instance, the popular Web Ontology Language OWL\footnote{https://www.w3.org/TR/owl2-overview/} is based on an expressive DL and 
large medical ontologies such as 
SNOMED\,CT\footnote{https://www.snomed.org/} and Galen\footnote{https://bioportal.bioontology.org/ontologies/GALEN}
have been developed using appropriate DLs.
A key feature of DLs is the ability to construct descriptions of complex concepts (i.e., sets of individuals sharing certain properties)
using concept names (unary predicates) and role names (binary predicates). For example, the concept of a parent can be described as
$\ex{Human}\sqcap\exists\ex{child}.\ex{Human}$. 
Knowledge about the relationship between concepts can then be expressed using concept inclusions (CIs), such as
$\ex{Human}\sqcap\exists\ex{child}.\ex{Human}\sqsubseteq\exists\ex{eligible}.\ex{TaxBreak}$, which says that parents are
eligible for a tax break.
 
Such purely qualitative statements are not always sufficient to express quantitative information (e.g. the number of children required for a tax break) that is relevant for an application domain.
\emph{Qualified number restrictions} that constrain the number of role successors belonging to a certain concept by a fixed natural number can be employed in DLs to express such quantitative information; 
e.g., $\ex{Human}\sqcap({\geq}\,3\,\ex{child}.\ex{Human})\sqsubseteq\exists\ex{eligible}.\ex{TaxBreak}$ says that a tax break is
available if one has at least three children.
\emph{Concrete domain restrictions} can represent a different type of quantitative information, where concrete objects such as numbers or strings
can be assigned to individuals using partial functions (\emph{features}). 
For example, a tax break might only be available if the annual salary is not too high. The CI
$\ex{Human}\sqcap({\geq}\,3\,\ex{child}.\ex{Human})\sqcap\exists\ex{salary}.{<}_{100,000} \sqsubseteq \exists\ex{eligible}.\ex{TaxBreak}$
specifies at least three children and an annual salary of less than 100,000\;\euro{} as eligibility criteria for a tax break.

Both (qualified) number restrictions \cite{DBLP:conf/ecai/HollunderNS90,DBLP:conf/kr/HollunderB91} and 
concrete domain restrictions~\cite{DBLP:conf/ijcai/BaaderH91} have been introduced early on in DL research,
but it turned out that they create considerable algorithmic challenges.
For \ALCQ, the extension of the basic DL \ALC with qualified number restrictions, it was open for a decade
whether the increase in expressivity also increases the complexity of reasoning if numbers in number restrictions are assumed to be represented
in binary, until Tobies \cite{DBLP:journals/logcom/Tobies01,DBLP:journals/jair/Tobies00} was able to show that it stays the same
(\PSpace without and \ExpTime with CIs). It also turned out that the unrestricted use of transitive roles within number restrictions
can cause undecidability \cite{HoSaTo00}. 
In~\cite{Baad17}, it was shown that reasoning in \ALCSCC, which extends \ALCQ with  very expressive counting constraints on role successors
expressed in the logic QFBAPA~\cite{KuRi07}, still has the same complexity as in \ALC and \ALCQ. In this logic, one can, e.g., 
describe humans that have exactly as many cars as children as $\ex{Human}\sqcap\suc(|\ex{own}\cap\ex{Car}|=|\ex{child}\cap\ex{Human}|)$,
without having to specify the exact numbers of cars and children. Such statements cannot even be expressed in full first-order logic~\cite{DBLP:conf/frocos/BaaderB19}.
 
The decidability result for \ALCD, i.e., \ALC extended with an admissible concrete domain \cDom, in~\cite{DBLP:conf/ijcai/BaaderH91} did not take 
CIs into account.
In the presence of CIs, integrating even rather simple concrete domains into the DL \ALC may cause 
undecidability~\cite{lutz2004nexp,DBLP:conf/cade/BaaderR20}.
In~\cite{LuMi07}, it was proved that integrating a so-called $\omega$-admissible concrete domain into $\ALC$ leaves reasoning decidable also in the 
presence of CIs. That paper gives two examples of such concrete domains (Allen's interval algebra~\cite{DBLP:journals/cacm/Allen83} and 
RCC8~\cite{DBLP:conf/kr/RandellCC92}). Using well-known notions and results from model theory, additional $\omega$-admissible concrete domains were 
exhibited in~\cite{DBLP:conf/cade/BaaderR20,BaRy22},
for example the rational numbers with comparisons $\mathfrak{Q}:=(\mathbb{Q},<,=,>)$.
Decidability results for \ALCD in the presence of CIs for concrete domains \cDom that are
not $\omega$-admissible can be found in~\cite{DBLP:conf/ecai/CarapelleT16,DBLP:conf/kr/LabaiOS20,DBLP:conf/jelia/DemriQ23}. 
A simpler, but considerably more restrictive way of achieving decidability is to use unary concrete domains~\cite{DBLP:conf/ijcai/HorrocksS01}.

In this paper, we study \ALCOSCCD, a combination of the DLs \ALCSCC and \ALCD with $\omega$-admissible concrete domains~\cDom as well as nominals ($\mathcal{O}$).
However, our logic goes beyond a pure combination of number restrictions and concrete domains by additionally allowing them to interact.
For a numerical concrete domain, it seems natural
to use the values of concrete features directly in the QFBAPA constraints, e.g.\ to describe people that own more
books than their age. We show, however, that this unrestricted combination easily leads to undecidability. Instead, we use concrete domain
constraints to define roles, which can then be employed within QFBAPA constraints. For example, the \emph{concrete role} $(\ex{salary} < \nx \ex{salary})$
connects an individual to all individuals that have a higher salary. One can use this to describe all persons that have a lower salary than at least half of their children with $\suc(|\mathsf{child} \cap (\mathsf{salary} < \nx \mathsf{salary})| > |\mathsf{child} \cap (\mathsf{salary} \ge \nx \mathsf{salary})|)$.
However, we show that the unrestricted use of such concrete roles also leads
to undecidability.
Hence, we additionally restrict them to pairs of individuals that are already connected by a role name.

Our main result is that the complexity of reasoning in \ALCOSCCD stays in \ExpTime if the complexity
of reasoning in \cDom is in \ExpTime. There are few results in the literature that determine the exact complexity of reasoning in DLs with concrete domains
\cite{DBLP:phd/dnb/Lutz02,DBLP:conf/kr/LabaiOS20,DBLP:conf/jelia/DemriQ23,DBLP:conf/dlog/BorgwardtBK24}. 
Only \cite{DBLP:phd/dnb/Lutz02} and \cite{DBLP:conf/dlog/BorgwardtBK24} consider $\omega$-admissible concrete domains, and the \ExpTime-completeness
result in the former is restricted to a specific temporal concrete domain. Our paper extends the results of the latter from \ALCD to \ALCOSCCD and is
generic since it holds for all $\omega$-admissible concrete domains with a decision problem in \ExpTime.
Finally, apart from the aforementioned undecidability results, we show that adding transitive roles also makes reasoning undecidable, even under strong syntactic restrictions.
All proof details can be found \inAppendix.

\section{Preliminaries}
\label{sec:preliminaries}

\paragraph{Concrete domains.}
We adopt the term \emph{concrete domain} to refer to a relational structure $\cDom=(D,P_1^D,P_2^D,\dots)$ over a non-empty, countable relational signature $\{P_1,P_2,\dots\}$, where $D$ is a non-empty set, and each predicate~$P_i$ has an associated arity $k_i \in \Nbb$ and is interpreted by a relation $P_i^D \subseteq D^{k_i}$.
An example is the structure $\stru{Q} := (\mathbb{Q},{<},{=},{>})$ over the rational numbers~$\mathbb{Q}$ with standard binary order and equality relations.
Given a countably infinite set~$V$ of variables, a \emph{constraint system} over $V$ is a set~$\cs$ of \emph{constraints} $P(v_1,\dotsc,v_k)$, where $v_1,\dotsc,v_k \in V$ and $P$ is a $k$-ary predicate of \cDom.
We denote by $V(\cs)$ the set of variables that occur in~$\cs$.
The constraint system $\cs$ is \emph{satisfiable} if there is a mapping $h \colon V(\cs) \to D$, called \emph{solution} of $\cs$, such that $P(v_1,\dotsc,v_k) \in \cs$ implies $(h(v_1),\dotsc,h(v_k)) \in P^D$.
The \emph{constraint satisfaction problem} for \cDom, denoted $\CSP(\cDom)$, asks if a given finite constraint system~$\cs$ over~\cDom is satisfiable.
The CSP of $\stru{Q}$ is decidable in polynomial time, by reduction to ${<}$-cycle detection: for example, the clique $x_1 < x_2,\ x_2 < x_3,\ x_3 < x_1$ is unsatisfiable over~$\stru{Q}$.

To ensure that reasoning in DLs with concrete domain restrictions remains decidable, we impose further properties on~\cDom regarding its predicates and the compositionality of its CSP for finite and countable constraint systems.
We say that \cDom is a \emph{patchwork} if it satisfies the following conditions:\footnote{Originally \cite{LuMi07} used \JEPD (jointly exhaustive, pairwise disjoint) and \AP (amalgamation property), while \JD (jointly diagonal) was later added by~\cite{BaRy22}.}
\begin{description}
  \item[\JEPD] if $k \ge 1$ and \cDom has $k$-ary predicates, then these predicates partition $D^k$;
  \item[\JD] there is a quantifier-free, equality-free first-order formula $\phi_{=}(x,y)$ over the signature of \cDom that defines equality between two elements of \cDom;
  \item[\AP] if $\ocs$, $\cs$ are constraint systems and
  $P(v_1,\dotsc,v_k) \in \ocs \;\text{iff}\; P(v_1,\dotsc,v_k) \in \cs$
  holds for all $v_1,\dotsc,v_k \in V(\ocs)\cap V(\cs)$ and all $k$-ary predicates $P$ over \cDom, then $\ocs$ and $\cs$ are satisfiable iff $\ocs \cup \cs$ is satisfiable.
\end{description}
If \cDom is a patchwork, we call a constraint system $\cs$ \emph{complete} if, for all $k \in \Nbb$ for which \cDom has $k$-ary predicates and all $v_1,\dotsc,v_k \in V(\cs)$, there is exactly one $k$-ary predicate~$P$ over \cDom such that $P(v_1,\dotsc,v_k) \in \cs$.
The concrete domain \cDom is \emph{homomorphism $\omega$-compact} if every countable constraint system $\cs$ over~\cDom is satisfiable whenever all its finite subsystems $\cs' \subseteq \cs$ are satisfiable.
We introduce \expadm concrete domains, which differ from $\omega$-admissible ones as defined in~\cite{LuMi07,BaRy22} by a stronger requirement on the decidability of $\CSP(\cDom)$.

\begin{definition}
  \label{dfn:omega-admissible}
  A concrete domain \cDom is \emph{\expadm} if it has a finite signature, it is a patchwork, it is homomorphism $\omega$-compact and its CSP is in \ExpTime.
\end{definition}
The finiteness of the signature of \cDom is necessary to ensure decidability.
Without this assumption, one can find instances of \cDom that satisfy all the other conditions of~\Cref{dfn:omega-admissible} such that reasoning in \ALCD is undecidable.
One such example is given by the concrete domain $(\Zbb, \{{+}_m \mid m \in \Zbb\})$ where ${+}_m$ relates those integers whose difference is equal to $m$~\cite{BaRy22}.
The conditions of \Cref{dfn:omega-admissible} are satisfied by Allen's interval algebra, RCC8 and $\mathfrak{Q}$~\cite{LuMi07,BaRy22}.
\paragraph{The logic QFBAPA.}
\emph{Set terms} are built using the operations intersection~$\cap$, union~$\cup$ and complement~$^c$ from \emph{set variables} and the constants~$\emptyset$ and~\Umc.
Set terms $s,t$ are then used in \emph{inclusion-} and \emph{equality constraints} $s \subseteq t$, $s = t$.
\emph{Presburger Arithmetic (PA) expressions} $\ell$, $\ell'$ of the form $n_0 + n_1 t_1 + \dotsb + n_k t_k$, where $n_i \in \Nbb$ and each $t_i$ is the cardinality $|s_i|$ of a set term $s_i$ or an integer variable, are used to form \emph{numerical constraints} $\ell = \ell'$, $\ell < \ell'$ and $n \dvd \ell$ ($n$ divides $\ell$), where $n\in\Nbb$.
A \emph{QFBAPA formula} is a Boolean combination of set- and numerical constraints.

A \emph{solution} $\sigma$ of a QFBAPA formula $\phi$ assigns a \emph{finite} set $\sigma(\Umc)$ to \Umc, subsets of $\sigma(\Umc)$ to set variables and integers to integer variables such that $\phi$ is \emph{satisfied} by~$\sigma$, which is defined in the standard way.
$\phi$ is \emph{satisfiable} if it has a solution.

The satisfiability problem for QFBAPA formulae is NP-complete.
Membership in NP is proved in~\cite{KuRi07}, using the notion of Venn regions.
If $\phi$ is a QFBAPA formula containing the set variables $X_1, \ldots, X_k$, a \emph{Venn region} for $\phi$ is a set term of the form
$X_1^{c_1} \cap \dotsb \cap X_k^{c_k}$
where each $c_i$ is either empty or the set complement~$c$.
For a Venn region $v$ for $\phi$ and a set variable $X$ in $\phi$, we write $X \in v$ to indicate that $X$ occurs without complement in $v$, and $X \notin v$ if $X^c$ occurs in $v$.
The following characterization, proved in~\cite{KuRi07} and strengthened in~\cite{Baad17}, guarantees the existence of solutions with a polynomial number of non-empty Venn regions for satisfiable QFBAPA formulae, as follows.
\begin{lemma}[\cite{Baad17}]
    \label{lem:preliminaries:savior}
    For every QFBAPA formula $\phi$, one can compute in polynomial time a number $N_\phi$, whose value is polynomial in the size of $\phi$, such that for every solution $\sigma$ of $\phi$ there exists a solution $\sigma'$ fulfilling the following conditions:
    \begin{itemize}
        \item there are at most $N_\phi$ Venn regions $v$ for $\phi$ for which $\sigma'(v) \ne \emptyset$;
        \item if $v$ is a Venn region for $\phi$ and $\sigma'(v) \ne \emptyset$, then $\sigma(v) \ne \emptyset$.
    \end{itemize}
\end{lemma}

\section{Syntax and Semantics of \texorpdfstring{\ALCOSCCD}{ALCOSCC(D)}}
\label{sec:dls}

We now introduce the classical description logic \ALCO, its extension \ALCOSCC, and finally our new logic \ALCOSCCD.

Given at most countable, disjoint sets \NC, \NR and \NI of \emph{concept}-, \emph{role}- and \emph{individual names}, \ALCO \emph{concepts} are built using negation~$\neg$ and conjunction~$\sqcap$ from concept names $A \in \NC$, \emph{nominals} $\{a\}$ with $a \in \NI$ and \emph{existential restrictions} $\exists r. C$ with $r \in \NR$ and $C$ an \ALCO concept~\cite{baader_horrocks_lutz_sattler_2017}. As usual, we use $C \sqcup D := \neg(\neg C \sqcap \neg D)$ (disjunction) and $\top := A \sqcup \neg A$.
An \emph{interpretation} \Imc consists of a \emph{domain} $\Delta^{\Imc} \ne \emptyset$ and a mapping ${\cdot}^{\Imc}$ assigning sets $A^{\Imc} \subseteq \Delta^{\Imc}$ to $A \in \NC$, relations $r^{\Imc} \subseteq \Delta^{\Imc} \times \Delta^{\Imc}$ to $r \in \NR$ and individuals $a^{\Imc} \in \Delta^{\Imc}$ to $a \in \NI$.
For $d \in \Delta^{\Imc}$, we define $r^{\Imc}(d) := \{ e \in \Delta^{\Imc} \mid (d,e) \in r^{\Imc} \}$.
We extend ${\cdot}^{\Imc}$ to concepts by
$(\lnot C)^\Imc := \Delta^\Imc \setminus C^\Imc$,
$(C\sqcap D)^\Imc := C^\Imc \cap D^\Imc$, $\{a\}^{\Imc} := \{a^{\Imc}\}$
and
$
    {(\exists r. C)}^{\Imc} := \{ d \in \Delta^{\Imc} \mid \exists e \in r^{\Imc}(d) \cap C^{\Imc}\}
$.
In this DL, the concept of all individuals that are human and have a child who is not \texttt{Sam} can be written as $\mathsf{Human} \sqcap \exists \mathsf{child}. \neg \{\mathtt{Sam}\}$.

\ALCOSCC extends \ALCO concepts with \emph{role successor restrictions} (or \emph{$\suc$-restrictions}) $\suc(\con)$, where \con is a set- or numerical constraint with role names and concepts as set variables and no integer variables, e.g.\ $r\subseteq C$~\cite{Baad17}.
This DL requires interpretations \Imc to be \emph{finitely branching}, i.e.\ such that the set of all role successors $\ars^{\Imc}(d) := \bigcup_{r\in\NR}r^\Imc(d)$ is finite, for all $d \in \Delta^{\Imc}$.
Then, each $d \in \Delta^{\Imc}$ induces a QFBAPA assignment~$\sigma_d$, where $\sigma_d(\Umc) := \ars^\Imc(d)$, $\sigma_d(r) := r^{\Imc}(d)$ for $r\in\NR$ and $\sigma_d(C) := C^{\Imc} \cap \ars^{\Imc}(d)$ for concepts~$C$.
The mapping $\cdot^\Imc$ is extended to $\suc$-restrictions by defining $d \in \suc(\con)^{\Imc}$ iff $\sigma_d$ is a solution of \con.

\ALCOSCC does not need existential restrictions $\exists r. C$, as they can be expressed as $\suc(|r \cap C| \ge 1)$.
Moreover, $\suc$-restrictions can compare quantities of successors, e.g.\ $\suc(|\ex{own}\cap\ex{Car}|=|\ex{child}\cap\ex{Human}|)$ describes people who own as many cars as they have children, without specifying the exact amount.

To integrate the concrete domain \cDom, we complement \NC, \NR and \NI with an at most countable set~\NF of \emph{feature names} that connect individuals with values in~$D$~\cite{DBLP:conf/ijcai/BaaderH91}.
A \emph{feature path} $p$ is of the form $f$ or $r f$ with $r \in \NR$ and $f \in \NF$. For instance, $\mathsf{salary}$ is a feature name as well as a feature path, while $\mathsf{child}\ \mathsf{salary}$ is a feature path including the role name~$\ex{child}$.
\emph{Concrete domain restrictions} (or \emph{CD-restrictions}) are concepts $\exists p_1,\dotsc,p_k. P$, where $p_i$ are feature paths and $P$ is a $k$-ary predicate of~\cDom.
An interpretation~\Imc assigns to each $f \in \NF$ a \emph{partial} function $f^{\Imc} \colon \Delta^{\Imc} \rightharpoonup D$.
A feature path~$p$ is mapped to $p^\Imc \subseteq \Delta^{\Imc} \times D$ by defining $p^{\Imc}(d) := \{ f^{\Imc}(d) \}$ if $p = f$\footnote{In a slight abuse of notation, we view $f^\Imc(d)$ both as a value and as a singleton set.} and $p^{\Imc}(d) := \{ f^{\Imc}(e) \mid e \in r^{\Imc}(d)\}$ if $p = r f$. 
Then we can define
\begin{equation*}
    (\exists p_1,\dotsc,p_k.P)^\Imc \!\coloneqq \big\{ d\in\Delta^\Imc\! \mid \text{\emph{some} tuple in}\;p_1^\Imc(d) \times \dotsb \times p_k^\Imc(d)\;\text{is in}\; P^D \}.
\end{equation*}
For example, one can describe individuals whose salaries are greater than that of some of their children using $\exists \mathsf{salary}, \mathsf{child}\ \mathsf{salary}. {>}$.
Furthermore, due to \JEPD, we can encode \emph{universal CD-restrictions} $\forall p_1,\dotsc,p_k. P$ using the conjunction of all concepts $\lnot\exists p_1,\dotsc,p_k. P'$ where $P' \ne P$ is a $k$-ary predicate of \cDom~\cite{LuMi07}.

A naive extension of \ALCOSCC with concrete domain reasoning that simply combines $\suc$- and CD-restrictions offers limited expressive power.
To improve that, we introduce \emph{feature pointers} $\alpha$ of the form $f$ or $\nx f$ with $f \in \NF$ and define \emph{feature roles} $\gamma := P(\alpha_1,\dotsc,\alpha_k)$, where each $\alpha_i$ is a feature pointer and $P$ is a $k$-ary predicate of \cDom.
For example, $\mathsf{salary}$ is a pointer to the salary of a given individual $d$, while $\nx \mathsf{salary}$ is a pointer to the salary of an individual~$e$ that we want to compare to $d$; the feature role $(\mathsf{salary} < \nx \mathsf{salary})$ describes a binary relation that contains $(d,e)$ iff the salary of $d$ is smaller than that of $e$.

We define \ALCOSCCD as the extension of \ALCOSCC with CD-restrictions and $\suc$-restrictions $\suc(\con)$ where \con can also contain feature roles as set variables.
We can now describe individuals that earn less than the majority of their children by
\begin{equation*}
    C_{\mathsf{ex}} := \suc(|\mathsf{child} \cap (\mathsf{salary} < \nx \mathsf{salary})| > |\mathsf{child} \cap (\mathsf{salary} < \nx \mathsf{salary})^c|).
\end{equation*}
Feature roles $\gamma := P(\alpha_1,\dotsc,\alpha_k)$ are mapped by interpretations~\Imc to relations $\gamma^{\Imc} \subseteq \Delta^{\Imc} \times \Delta^{\Imc}$ such that $(d,e) \in \gamma^{\Imc}$ iff $(c_1,\dots,c_k)\in P^D$, where $c_i:=f_i^\Imc(d)$ if $\alpha_i = f_i$ and $c_i:=f_i^\Imc(e)$ if $\alpha_i = \nx f_i$.
The QFBAPA assignment $\sigma_d$ is extended to map feature roles $\gamma$ to $\gamma^{\Imc} \cap \ars^{\Imc}(d)$, and $\suc(\con)^\Imc$ is defined as before.

An \emph{\ALCOSCCD TBox} \Tmc is a finite set of \emph{concept inclusions (CIs)} $C \sqsubseteq D$ between concepts $C,D$.
For example, we can describe an individual \texttt{Jane} that earns more than \texttt{Sam}, where the role $\mathsf{ref}_{\mathtt{Sam}}$ always points to \texttt{Sam}:
\begin{equation*}
    \Tmc_{\mathsf{ex}} := \big\{\, \top \sqsubseteq \suc(\mathsf{ref}_\mathtt{Sam} = \{\mathtt{Sam}\}), \; \{\mathtt{Jane}\} \sqsubseteq \exists \mathsf{salary}, \mathsf{ref}_{\mathtt{Sam}}\ \mathsf{salary}. {>} \,\big\}.
\end{equation*}
A finitely branching interpretation \Imc is a \emph{model} of~\Tmc if $C^{\Imc} \subseteq D^{\Imc}$ holds for every CI $C \sqsubseteq D$ in~\Tmc.
A TBox \Tmc is \emph{consistent} if it has a model.

We mentioned above that feature roles make \ALCOSCCD more expressive. Precisely, we can show that some concepts with feature roles cannot be equivalently expressed by only using feature names in CD-restrictions; two concepts are \emph{equivalent} if they are always interpreted by the same sets of individuals.
\begin{theorem}
    There is no $\ALCOSCC(\qDom)$ concept without feature roles that is equivalent to $C_{\mathsf{ex}}$.
\end{theorem}
\begin{proof}
    Assume that there is an $\ALCOSCC(\qDom)$ concept $D$ without feature roles that is equivalent to $C_{\mathsf{ex}}$.
    Consider the interpretation \Imc, where $d$ has $\mathsf{salary}$~$0$ and five $\mathsf{child}$-successors, two whose $\mathsf{salary}$ is~$0$ and three whose $\mathsf{salary}$ is~$1$.
    Then, $d \in C_{\mathsf{ex}}^{\Imc} = D^{\Imc}$.
    Construct \Jmc from~\Imc by changing the $\mathsf{salary}$ of one of the successors from~$1$ to~$0$.
    Since every individual in \Jmc satisfies the same CD-restrictions, concept names and $\suc$-restrictions without feature roles as in~\Imc, we deduce that $d \in D^{\Jmc}$.
    However, $d$ has more successors with equal $\mathsf{salary}$ in \Jmc than successors with greater $\mathsf{salary}$, hence $d \notin C_{\mathsf{ex}}^{\Jmc} = D^{\Jmc}$ must hold, which is a contradiction.
    \qed
\end{proof}

\section{Deciding Consistency}
\label{sec:satisfiability}

Let \cDom be an \expadm concrete domain and \Tmc an \ALCOSCCD TBox.
In this section, we assume w.l.o.g.\ that \NC, \NR, \NI and \NF are finite and contain exactly the names occurring in~$\Tmc$ and that there is at least one individual name; indeed, \Tmc is consistent iff $\Tmc \cup \{ \{a\} \sqsubseteq \{a\} \}$ is consistent, where $a$ is a fresh individual name.
We define the notion of \emph{individual types}, describing sets of equivalent individual names in an interpretation.
\begin{definition}
  \label{dfn:individual-type}
  An \emph{individual type} $\ain$ w.r.t. \NI is a non-empty subset of \NI, and a set of individual types \Ibb is an \emph{individual type system} for \NI if \Ibb partitions \NI.
  Given an interpretation~\Imc, an individual $d \in \Delta^{\Imc}$ \emph{has individual type} $\ain_{\Imc}(d) := \{ a \in \NI \mid a^{\Imc} = d\}$ if this set is non-empty, and $d$ is \emph{anonymous} otherwise.
\end{definition}
We now fix an individual type system~\Ibb.
Let \Mmc be the set of all subconcepts appearing in~\Tmc, as well as their negations.
We define the notion of \emph{type} as usual.

\begin{definition}
  \label{dfn:type}
  A \emph{type} w.r.t.~\Tmc is a set $t \subseteq \Mmc$ such that:
  \begin{itemize}
    \item if $C\sqsubseteq D\in\Tmc$ and $C \in t$, then $D \in t$;
    \item if $\neg C \in \Mmc$, then $C \in t$ iff $\neg C \notin t$;
    \item if $C \sqcap  C' \in \Mmc$, then $C \sqcap C' \in t$ iff $C \in t$ and $C' \in t$.
  \end{itemize}
  The \emph{type} of $d \in \Delta^{\Imc}$ w.r.t. \Tmc is the set
  $t_{\Imc}(d) := \big\{ C \in \Mmc \mid d \in C^{\Imc} \big\}$.
\end{definition}
If \Imc is a model of \Tmc, then $t_{\Imc}(d)$ is indeed a type w.r.t.~\Tmc.
A type $t$ is \emph{named with} an individual type~$\ain_t$ if for all $a \in \NI$, $a \in \ain_t$ iff $\{a\} \in t$, and is \emph{anonymous} if it is not named with any individual type.

Following the approach used in~\cite{Baad17}, we construct a QFBAPA formula $\phi_t$ that is induced by the $\suc$-restrictions $\suc(\con)$ in a type~$t$ and enriched with constraints derived from the individual type system~\Ibb and the set of role names~\NR.
Formally, $\phi_t$ is defined as the conjunction of
\begin{itemize}
  \item $\phi_\con$ if $\suc(\con)\in t$ and $\neg \phi_{\con}$ otherwise, where $\phi_\con$ is derived from $\con$ by replacing role names~$r$, feature roles~$\gamma$ and concepts~$C$ with set variables $X_r$, $X_{\gamma}$ and $X_C$, respectively;
  \item $|\bigcap_{a \in \ain} X_{\{a\}}|\le 1$ for every $\ain \in \Ibb$; and
  \item $\Umc=\bigcup_{r\in\NR}X_r$.
\end{itemize}
All formulae $\phi_t$ contain exactly the same set variables and thus have the same Venn regions (cf.\ \Cref{sec:preliminaries}), called the \emph{Venn regions of~\Tmc}.
A Venn region $v$ of \Tmc \emph{has individual type} $\ain_v = \{ a \in \NI \mid X_{\{a\}} \in v\}$ if this set is non-empty, and $v$ is \emph{anonymous} otherwise.
The following example shows that $\phi_t$ does not yet account for the numerical constraints induced by the CD-restrictions in $t$.
\begin{example}
  \label{exa:unsatisfiable-lower-bounds}
  Let $\Tmc=\{\, \top \sqsubseteq (\exists \mathsf{salary}, \mathsf{child}\ \mathsf{salary}.{<}) \sqcap (\suc(|\mathsf{child}| \le 0)) \,\}$.
  For every model \Imc of \Tmc and $d \in \Delta^{\Imc}$, the type $t := t_{\Imc}(d)$ contains both conjuncts appearing in this CI.
  The QFBAPA formula $\phi_t := |X_{\mathsf{child}}| \le 0 \land \Umc = X_{\mathsf{child}}$ is satisfied by any solution assigning the empty set to $\Umc$.
  However, $t$ cannot be realized: the first conjunct implies that $d$ has a $\mathsf{child}$-successor $e \ne d$ such that $\mathsf{salary}^{\Imc}(d) < \mathsf{salary}^{\Imc}(e)$, while the last conjunct forces $d$ to have no $\mathsf{child}$-successor.
\end{example}
To realize the CD-restrictions in~$t$, we may need up to $\mt:=R_\Tmc\cdot P_\Tmc$ distinct role successors, where $R_\Tmc$ is the number of CD-restrictions in \Mmc and $P_\Tmc$ is the maximal arity of predicates of \cDom occurring in~\Mmc.
We add this information to the QFBAPA formula $\phi_t$ with additional constraints over a set of pre-selected Venn regions, representing sets of role successors whose existence is implied by the CD-restrictions in~$t$.
Let $S$ be a set of at most $\mt$ Venn regions $v$, each associated to a natural number $0 \le n_v \le \mt$.
By~\Cref{lem:preliminaries:savior}, the QFBAPA formula $\phi_{t,S}$, which extends $\phi_t$ with a conjunct $|v| \ge n_v$ for each $v \in S$, is satisfiable iff there is a natural number \nt of polynomial size w.r.t.\ the size of~$\phi_t$ and \mt s.t.\ $\phi_{t,S}$ has a solution in which at most \nt Venn regions are non-empty.
Moreover, since all formulae~$\phi_t$ are nearly of the same size (except for the difference between $\phi_\con$ and $\lnot\phi_\con$) and $|S|$ and the numbers~$n_v$ are bounded by~$M_\Tmc$, we can assume that the bound \nt is independent of the choice of $S$ and $t$, is polynomial w.r.t.\ the size of \Tmc and can be computed in polynomial time.

To formalize these additional restrictions, we consider \emph{bags}, i.e.\ functions~$V$ assigning to every Venn region $v$ of \Tmc a \emph{multiplicity} $V(v) \in \Nbb$, whose \emph{support} $\supp(V)$ is the set of Venn regions of \Tmc with multiplicity $V(v) \ge 1$.
The associated QFBAPA formula $\phi_V$ is the conjunction of the constraint $\Umc = \bigcup \supp(V)$ and all constraints $|v| \ge c$ where $v \in \supp(V)$ and $c = V(v)$.
\begin{definition}\label{dfn:venn-bag}
  A \emph{Venn bag} for a type~$t$ w.r.t.~\Tmc is a bag~$V$ of Venn regions of~\Tmc s.t. $|\supp(V)| \le \nt$, $V(v) \le \mt + 1$ holds for all $v \in \supp(V)$ and the QFBAPA formula $\phi_{t,V} := \phi_t \land \phi_V$ is satisfiable.
\end{definition}
By Lemma~\ref{lem:preliminaries:savior}, $\phi_{t,S}$ is satisfiable iff there is a Venn bag $V$ for $t$ such that $\phi_{t,V}$ includes all constraints from $\phi_{t,S}$.

Finally, we take care of actually satisfying the CD-restrictions occurring in a type by using complete constraint systems to describe all relevant feature values.
Feature values that are not represented in these systems correspond to undefined values.
To ensure that all types agree on the feature values of individual names, we fix an \emph{individual constraint system} \csi w.r.t.~\Ibb, i.e.\ a complete constraint system over variables of the form $f^{\ain}$, where $f\in\NF$ and $\ain \in \Ibb$, that refer to the feature values of named individuals.
Then, we define constraint systems $\cs_{t,V}$ representing the relations between the feature values associated with an individual of type~$t$ and those of its role successors as specified by a Venn bag~$V$ for~$t$.
The system $\cs_{t,V}$ extends $\csi$ by adding variables of the form
\begin{itemize}
  \item $f^{\star}$, representing the value of the feature~$f\in\NF$ at the current individual;
  \item $f^{(v,j)}$ with $v \in \supp(V)$ and $1 \le j \le V(v)$ for the $f$-values at the successors, in order to express the relevant CD-restrictions.
\end{itemize}
Again, not all these variables actually need to occur in the constraint system, only the ones whose associated feature values should be defined.
To handle named types and named Venn regions, we define the indexing functions
\[
  \iota(t) := \begin{cases}
    \star & \text{if $t$ is anonymous} \\
    \ain_t & \text{otherwise}
  \end{cases}
  \qquad
  \text{and }\iota((v,j)) := \begin{cases}
    (v,j) & \text{if $v$ is anonymous} \\
    \ain_v & \text{otherwise}
  \end{cases}
\]
for all $v \in \supp(V)$ and $1 \le j \le V(v)$.
Additionally, we do not allow more variables of the form~$f^\ain$ than those already contained in~\csi.
\begin{definition}
  \label{dfn:local-system}
  Let $t$ be a type w.r.t.~\Tmc and $V$ a Venn bag for~$t$.
  A \emph{local system} for $t,V$
  is a complete constraint system $\cs_{t,V}$ that includes $\csi$ and no additional variables of the form $f^\ain$, $\ain\in\Ibb$, such that:
  \begin{enumerate}
    \item if $C := \exists p_1,\dots,p_k.P\in\Mmc$, then
    $C \in t$ iff $P(f_1^{x_1},\dotsc,f_k^{x_k}) \in \cs_{t,V}$ such that 
    \begin{equation*}
      x_{i} = 
      \begin{cases}
        \iota(t) &\text{if $p_{i}={f_i}$, or} \\
        \iota((v,j)) &\text{if $p_{i}=rf_{i}$, for some $1 \le j \le V(v)$ and $X_r \in v$};
      \end{cases}
    \end{equation*}
    \item for all set variables $X_{P(\alpha_1,\dots,\alpha_k)}$, all $v\in\supp(V)$, and $1\le j\le V(v)$ it holds that $X_{P(\alpha_1,\dots,\alpha_k)}\in v$ iff $P(f_1^{x_1},\dotsc,f_k^{x_k}) \in \cs_{t,V}$, where
    \begin{equation*}
      x_{i} = 
      \begin{cases}
        \iota(t) &\text{if $\alpha_{i}=f_i$, and} \\
        \iota((v,j)) &\text{if $\alpha_{i}=\nx f_i$}.
      \end{cases}
    \end{equation*}
  \end{enumerate}
\end{definition}
In the following definition, we denote with $S_v$ the subset of \Mmc that contains $C \in \Mmc$ if $X_C \in v$ and $\neg C \in \Mmc$ if $X_C \notin v$ (cf.\ \Cref{sec:preliminaries}).

\begin{definition}          
  \label{dfn:augmented-type}
  An \emph{augmented type} for~\Tmc is a tuple $\agt := (t,V,\cs_\agt)$, where $t$ is a type w.r.t.~\Tmc, $V$ is a Venn bag for~$t$, and $\cs_\agt$ is a satisfiable local system for~$t,V$.
  The \emph{root} of~$\agt$ is $\agroot(\agt) := t$.

  An augmented type $\agt'=(t',V',\cs_{\agt'})$ \emph{patches} $\agt$ at $(v,i)$, where $v\in\supp(V)$ and $1 \le i \le V(v)$, if $S_v \subseteq t'$ and the \emph{merged system} $\cs_{\agt} \merge{(v,i)} \cs_{\agt'}$ has a solution, where $\cs_{\agt} \merge{(v,i)} \cs_{\agt'}$ is obtained as the union of~$\cs_{\agt}$ and the result of replacing all variables in $\cs_{\agt'}$ as follows:
  \[\begin{array}{lll@{\quad}l}
    f^\star & \mapsto & f^{(v,i)} & \text{if $t'$ is anonymous}; \\
    f^{(w,j)} & \mapsto & {f^{(w,j)}}' & \text{for all anonymous $w\in\supp(V')$ and $1\le j\le V'(w)$}; \\
    f^\ain & \mapsto & f^\ain & \text{for all $\ain\in\Ibb$}.
  \end{array}\]
  A set of augmented types \Tbb \emph{patches} $\agt$ if, for all $v \in \supp(V)$ and $1 \le i \le V(v)$, there is a $\agt' \in\Tbb$ that patches $\agt$ at $(v,i)$.
\end{definition}
The merging operation identifies all features associated to $(v,i)$ in~$\cs_{\agt}$ with those associated to $t'$ in~$\cs_{\agt'}$, while keeping the remaining variables associated to anonymous individuals separate.
If $t'$ is not anonymous (and thus $\cs_{\agt'}$ contains no variable of the form $f^\star$) then the condition $S_v \subseteq t'$ ensures that $\ain_v = \ain_{\agt'}$, and thus the variable $f^{\iota((v,i))} = f^{\ain_v} = f^{\ain_{\agt'}}$ in $\cs_\agt$ is already identical to $f^{\iota(t')} = f^{\ain_{\agt'}}$ in~$\cs_{\agt'}$.

The augmented types are now used by \Cref{alg:type-elimination} to decide consistency of an \ALCOSCCD TBox via a type elimination approach. We show that \Cref{alg:type-elimination} is indeed sound and complete.

\begin{algorithm}[tb]
  \caption{Type elimination algorithm for \ALCOSCCD}
  \label{alg:type-elimination}
  \begin{algorithmic}[1]
    \Require An \ALCOSCCD TBox \Tmc
    \Ensure \textsc{consistent} if \Tmc is consistent, and \textsc{inconsistent} otherwise
    \State \textbf{guess} an individual type system \Ibb and an individual constraint system \csi
    \State \textbf{guess} augmented types $\agt_{\ain}=(t_{\ain},V_{\ain},\cs_{\ain})$ for $\ain \in \Ibb$ s.t. $t_\ain$ is named with $\ain$
    \State $\Tbb \gets \{ \agt = (t,V,\cs)\;\text{augmented type} \mid t \;\text{is anonymous} \} \cup \{ \agt_{\ain} \mid \ain \in \Ibb \}$
    \While{there is $\agt\in\Tbb$ that is not patched by \Tbb}
      $\Tbb\gets\Tbb\setminus\{t\}$
    \EndWhile
    \If{$\agt_{\ain}\in\Tbb$ for all $\ain \in \Ibb$}
      \Return{\textsc{consistent}}
    \Else{}
      \Return{\textsc{inconsistent}}
    \EndIf
  \end{algorithmic}
\end{algorithm}

\begin{lemma}\label{lem:soundness}
  If there is a run of \Cref{alg:type-elimination} that returns \textsc{consistent}, then \Tmc is consistent.
\end{lemma}
\begin{proof}
  We construct a model $\Imc$ of $\Tmc$
  using the individual type system \Ibb and the set \Tbb of augmented types constructed by \Cref{alg:type-elimination}.
  The domain $\Delta^{\Imc}$ consists of tuples $(\ain,w)$, where $\ain \in \Ibb$ and $w$ is a word over the alphabet $\Sigma$ of all tuples $(\agt, v, i)$ with $\agt \in \Tbb$, $v$ a Venn region of \Tmc and $i \ge 1$ a natural number.
  We associate to each tuple $(\ain,w)$ the augmented type $\wend(\ain,w) \in \Tbb$ defined as $\wend(\ain,\varepsilon) := \agt_{\ain}$ and $\wend(\ain,w' \cdot (\agt,v,i)) := \agt$ for $w' \in \Sigma^{*}$.
  
  We define $\Delta^{\Imc}$ as the union of sets $\Delta^m$ with $m \in \Nbb$, where $\Delta^0$ contains $(\ain,\varepsilon)$ for every $\ain \in \Ibb$ and $\Delta^{m+1}$ is defined inductively in the following.
  Given $(\ain,w) \in \Delta^m$ with $\wend(\ain,w) = \agt = (t,V,\cs_{t,V}) \in \Tbb$ we observe that
  \begin{itemize}
    \item the QFBAPA formula $\phi_{t}$ has a solution $\sigma_{\ain,w}$ such that $\sigma_{\ain,w}(|v|) \ge V(v)$ if $v \in \supp(V)$ and $\sigma_{\ain,w}(|v|) = 0$ otherwise for all Venn regions $v$ of~\Tmc,
    \item for $v \in \supp(V)$ and $i = 1,\dotsc,V(v)$ there exists an augmented type $\agt_{(v,i)}\in \Tbb$ patching $\agt$ at $(v,i)$, as otherwise $\agt$ would have been eliminated from \Tbb.
  \end{itemize} 
  Using these augmented types, we define for $r \in \NR$ the set $\Delta_r^{m+1}[\ain,w]$ containing $(\ain,w \cdot (\agt_{(v,i)},v,j))$ iff $X_r \in v$, $\agroot(\agt_{(v,i)})$ is anonymous, $j = 1,\dotsc,\sigma_{\ain,w}(|v|)$ and $i = \max(j,V(v))$.
  We now define $\Delta^{m+1}$ as the extension of $\Delta^{m}$ by all sets $\Delta_r^{m+1}[\ain,w]$ for which $r \in \NR$ and $(\ain,w) \in \Delta^m$.

  The extensions of $a \in \NI$, $A \in \NC$ and $r \in \NR$ are defined as:
  \begin{align*}
    a^{\Imc} &:= (\ain,\varepsilon)\text{, where $\ain$ is the unique individual type in $\Ibb$ containing $a$;} \\
    A^{\Imc} &:= \{ (\ain,w) \in \Delta^{\Imc} \mid \wend(\ain,w) = \agt \;\text{and}\; A \in \agroot(\agt) \}; \\
    r^{\Imc} &:= \{ ((\ain,w),(\mathfrak{b},\varepsilon)) \mid  \wend(\ain,w) = \agt \;\text{and $\agt_{\mathfrak{b}}$ patches $\agt$ at $(v,i)$ with $X_r\in v$} \} \cup{} \\
    &\phantom{{}:={}} \{ ((\ain,w),(\ain,w')) \mid (\ain,w) \in \Delta^m \;\text{and}\; (\ain,w') \in \Delta_r^{m+1}[\ain,w] \;\text{with}\; m \in \Nbb \}.
  \end{align*}
  For $f \in \NF$, $f^{\Imc}$ is defined as follows.
  If $(\ain,w) \in \Delta^{\Imc}$ and $\wend(\ain,w) = (t,V,\cs_{t,V})$, we extend $\cs_{t,V}$ with all variables $f^{(v,i)}$ where $i > V(v)$ and $f \in \NF$ such that $f^{(v, V(v))}$ occurs in $\cs_{t,V}$ and $(\ain, w \cdot (\agt',v,i))$ occurs in $\Delta^{\Imc}$.
  Then, we add all constraints $P((f_1^{x_1})',\dotsc,(f_k^{x_k})')$ obtained by replacing every occurrence of $f^{(v,V(v))}$ in a constraint $P(f_1^{x_1},...,f_k^{x_k}) \in \cs_{t,V}$ with some variable $f^{(v,i)}$ among those occurring in the extended system with $i \ge V(v)$.
  In this way, the feature values of all role successors of $(\ain,w)$ are handled correctly w.r.t. one another and w.r.t. those of $(\ain,w)$.
  Next, we replace all variables $f^{\ain}$ with $f^{\ain,\varepsilon}$, all variables $f^{\star}$ with $f^{\ain,w}$ and all variables $f^{(v,i)}$ with $f^{\ain,u}$ where $u$ is the unique word of the form $w \cdot (\agt',v,i)$ for which $(\ain, u) \in \Delta^{\Imc}$.

  Let $\cs_{\ain,w}$ be the resulting complete constraint system and $\cs^m$ with $m \in \Nbb$ be the union of all $\cs_{\ain,w}$ with $(\ain,w) \in \Delta^m$. We show \inAppendix (\Cref{lem:soundness-features}) that
  for every $m\in\Nbb$, the constraint system $\cs^m$ has a solution.
  Let $\cs^{\Imc}$ be the union of all systems $\cs^m$ with $m \in \Nbb$.
  Every finite subsystem of $\cs^{\Imc}$ is a subsystem of $\cs^m$ for some $m \in \Nbb$ and is thus satisfiable.
  Thus, by homomorphism $\omega$-compactness of \cDom, we can infer that $\cs^{\Imc}$ has a solution $h^{\Imc}$.
  For every $f\in\NF$ and $(\ain,w) \in \Delta^{\Imc}$, we now define $f^{\Imc}((\ain,w)) := h^{\Imc}(f^{\ain,w})$ if $f^{\ain,w}$ occurs in $\cs^{\Imc}$ and leave it undefined otherwise.

  We show \inAppendix (\Cref{lem:soundness-concepts}) that for all $d = (\ain,w) \in \Delta^{\Imc}$ and $C \in \Mmc$, we have $C \in \agroot(\wend(d))$ iff $d \in C^{\Imc}$.
  It is a direct consequence that $\Imc$ satisfies all CIs in \Tmc and is thus a model of $\Tmc$.
  Hence, \Tmc is consistent.
  \qed
\end{proof}

\begin{lemma}\label{lem:completeness}
  If \Tmc is consistent, then there is a run of \Cref{alg:type-elimination} that returns \textsc{consistent}.
\end{lemma}
\begin{proof}
  Let \Imc be a model of \Tmc and \Ibb be the individual type system that contains an individual type~$\ain$ iff $\ain = \ain_{\Imc}(d)$ for some $d \in \Delta^{\Imc}$  (cf.\ \Cref{dfn:individual-type}).
  Then \Ibb is well-defined, as every $a \in \NI$ is uniquely assigned to $a^\Imc \in \Delta^{\Imc}$.
  For $\ain\in\Ibb$, we denote by $\ain^\Imc$ the unique $d\in\Delta^\Imc$ with $\ain=\ain_\Imc(d)$.
  Further, let $T_\Imc:=\{t_\Imc(d) \mid d\in\Delta^\Imc \}$ be the set of all types that are realized in~\Imc.

  For each individual type $\ain \in \Ibb$, we define $t_\ain := t_{\Imc}(\ain^\Imc)$.
  Using \Ibb and $T_{\Imc}$, we build a constraint system $\csi$ and a set $\Tbb_{\Imc} := \{ \agt_{\Imc}(d) \mid d \in \Delta^{\Imc}\}$ of augmented types, containing unique types $\agt_{\ain}$ whose roots are named with $\ain \in \Ibb$.
  We define the individual constraint system $\csi$ over all variables $f^\ain$ with $f\in\NF$ and $\ain\in\Ibb$ for which $f^\Imc(\ain^\Imc)$ is defined, such that $P(f_1^{\ain_1},\dots,f_k^{\ain_k}) \in \csi$ iff $\big(f_1^\Imc(\ain_1^\Imc),\dots,f_k^\Imc(\ain_k^\Imc)\big)\in P^D$.
  Clearly, $\csi$ is complete.

  Next, we associate to each $d \in \Delta^{\Imc}$ an augmented type $\agt_{\Imc}(d) := (t_{\Imc}(d),V_d,\cs_d)$.
  If $e$ is a role successor of $d$, let $v_e$ be the Venn region of \Tmc whose variables $X_r$, $X_C$, $X_\gamma$ for role names~$r$, concepts~$C$ and feature roles~$\gamma$ satisfy the following:
  \begin{itemize}
    \item $X_r \in v_e$ iff $e$ is an $r$-successor of $d$; 
    \item $X_C \in v_e$ iff $C \in t_\Imc(e)$;
    \item $X_\gamma \in v_e$ iff $e \in \gamma^{\Imc}(d)$.
  \end{itemize}
  For every non-negated CD-restriction $\exists p_1,\dots,p_k.P\in t_\Imc(d)$, we can find values $c_i \in p_i^{\Imc}(d)$ for $i = 1,\dotsc,k$ such that $(c_1,\dotsc,c_k) \in P^D$.
  If $p_i=r_if_i$, this implies that there is $e_i \in r_i^{\Imc}(d)$ such that $f_i^\Imc(e_i)=c_i$.
  We collect all these successors of~$d$, which are at most \mt many distinct elements, in the set $\wit$.
  For $e \in \wit$, let $n_e$ be the number of elements $e' \in \wit$ such that $v_{e} = v_{e'}$.
  We show \inAppendix (\Cref{claim:completeness-bag}) that there is a Venn bag $V_d$ for $t_{\Imc}(d)$ w.r.t.~\Tmc such that $V_d(v_e) = n_e$ for all $e \in \wit$, and for all other $v\in\supp(V_d)$ we have $V_d(v)=1$ and there is a role successor $e\in\Delta^\Imc\setminus\wit$ of~$d$ such that $v=v_e$.

  It remains to define the local system~$\cs_d$.
  Consider the set~$X_d$ that contains $\iota(t_\Imc(d))$ (either $\star$ or $\ain_{t_\Imc(d)}$), all $\ain\in\Ibb$, and all pairs $(v,j)$ with anonymous Venn bags $v\in\supp(V_d)$ and $1\le j\le V_d(v)$ (cf.\ \Cref{dfn:local-system}).
  Let $\lambda_d$ be a bijection mapping every such $(v,j)$ to an anonymous successor~$e$ of $d$ satisfying $v_e = v$, such that $\lambda_d((v_e,j))\in\wit$ for all $e\in\wit$.
  Such a bijection exists due to \Cref{claim:completeness-bag}.
  We extend this bijection to~\Ibb by setting $\lambda_d(\ain):=\ain^\Imc$.
  Furthermore, we extend~$\lambda_d$ to~$\xi_d$ with $\xi_d(\iota(t_\Imc(d))):=d$.
  Then, $\xi_d$ is injective except for $\star$: if $d$ is its own role successor, it can happen that $\xi_d(\star)=d=\xi_d((v,j))$.
  We define the complete constraint system $\cs_d$ over variables $f^x$ with $f \in \NF$, $x \in X_d$, such that $P(f_1^{x_1},\dotsc,f_k^{x_k}) \in \cs_d$ iff $\big(f_1^\Imc(\xi_d(x_1)),\dots,f_k^\Imc(\xi_d(x_k))\big)\in P^D$ holds for all $x_1,\dotsc,x_k \in X_d$.
  If $f^\Imc(\xi_d(x))$ is undefined, then $f^x$ does not occur in~$\cs_d$.

  We show \inAppendix (\Cref{claim:completeness-local-system}) that $\cs_d$ is a satisfiable local system for $t_{\Imc}(d)$ and~$V_d$.
  Thus, $\mathfrak{t}_\Imc(d) = \big( t_\Imc(d), V_d, \cs_d \big)$ is an augmented type.
  Furthermore, we can show (cf.\ \Cref{claim:completeness-patching} \inAppendix) that every augmented type in $\Tbb_\Imc$ is patched in $\Tbb_\Imc$.
  If \Cref{alg:type-elimination} guesses $\Ibb$, $\csi$, and $\agt_\ain$, where $\ain\in\Ibb$, then the initial set \Tbb must contain $\Tbb_\Imc$, and no augmented type from $\Tbb_\Imc$ can ever be removed from~\Tbb. This shows that the augmented types $\{ \agt_\ain \mid \ain\in\NI \} \subseteq \Tbb_\Imc$ remain in \Tbb throughout the execution of the algorithm. Since the algorithm terminates, it thus returns \textsc{consistent}.
  \qed
\end{proof}
Because \Cref{alg:type-elimination} runs in exponential time, we obtain a matching upper bound to the \ExpTime-hardness inherited from $\ALC$. Indeed, as there are at most exponentially many individual type systems and polynomially many individual types in such a type system, all guesses can be implemented by enumerating all choices in exponential time. The main elimination procedure also runs in exponential time as the number of augmented types is exponentially bounded, and all required tests can be performed in exponential time, provided that $\cDom$ is \ExpTime-$\omega$-admissible. We thus obtain the following result.
\begin{restatable}{theorem}{ThmDecidability}
  \label{thm:decidability}
  Let \cDom be an \ExpTime-$\omega$-admissible concrete domain.
  Then, consistency checking in \ALCOSCCD is \ExpTime-complete.
\end{restatable}

\section{Reasoning with ABoxes}
\label{sec:constants}

In DLs, a TBox is often complemented by an \emph{ABox} containing \emph{concept assertions} $C(a)$ and \emph{role assertions} $r(a,b)$, where $a,b\in\NI$, $r\in\NR$, and $C$ is a concept, with the obvious semantics.
In our DL, those assertions can be expressed in the TBox using nominals~\cite{baader_horrocks_lutz_sattler_2017}.
In the presence of a concrete domain, however, we may want to use additional kinds of assertions: \emph{predicate assertions} $P(f_1(a_1),\dotsc,f_k(a_k))$ with $f_i \in \NF$, $a_i \in \NI$, $i = 1,\dotsc,k$, and a $k$-ary predicate~$P$ of~\cDom, and \emph{feature assertions} $f(a,c)$ with $f\in\NF$, $a\in\NI$, and a constant $c\in D$.
The former requires every model~\Imc to satisfy $(f_1^{\Imc}(a_1^{\Imc}),\dotsc,f_k^{\Imc}(a_k^{\Imc})) \in P^D$, and the latter states that $f^\Imc(a^\Imc)=c$.

Using predicate assertions, we can rewrite the TBox $\Tmc_{\mathsf{ex}}$ in~\Cref{sec:dls} into a single, intuitive assertion $\ex{salary}(\texttt{Sam}) < \ex{salary}(\texttt{Jane})$.
This also demonstrates how predicate assertions can be simulated by CIs: instead of $P(f_1(a_1),\dotsc,f_k(a_k))$, we can use $\top \sqsubseteq \suc(\mathsf{ref}_{a_i} = \{a_i\})$ for $i = 1,\dotsc,k$ and $\top \sqsubseteq \exists \mathsf{ref}_{a_1} f_1, \dotsc, \mathsf{ref}_{a_k} f_k. P$.

On the other hand, with feature assertions, we can give specific values and state, for instance, that \texttt{Sam}'s salary is 100,001\;\euro{} with $\ex{salary}(\texttt{Sam},\,\text{100,001})$.
Feature assertions seemingly increase the expressivity, since we can use them to refer to constant values.
However, we first have to specify how these constants are actually encoded.
For the following results, we consider \emph{concrete domains~\cDom with constants}, which extend concrete domains with an encoding for arbitrary constants $c\in D$ and constraint systems that can use such constants in addition to variables.
For \expadm concrete domains with constants, we also require the extended CSP(\cDom) to be decidable in exponential time.
In particular, the main known examples of \expadm concrete domains ($\stru{Q}$, Allen's relations, and RCC8) satisfy this requirement under the reasonable assumptions that all numbers are given as integer fractions in binary encoding and the constants in RCC8 refer to polygonal regions in the rational plane~\cite{Jo18,LiLiWa13}.
We can now use this encoding to represent the constants in feature assertions $f(a,c)$.

Unfortunately, we cannot directly use constants in constraints to extend \Cref{alg:type-elimination} to support feature assertions, since this would result in infinitely many possible local constraint systems.
Another idea to deal with feature assertions is that,
if \cDom has \emph{singleton predicates}~$=_c$ with $(=_c)^D = \{c\}$, then one can express $f(a,c)$ by $\{a\} \sqsubseteq \exists f.{=_c}$.
Since an $\omega$-admissible concrete domain \cDom has a finite signature, however, this only works for a fixed, finite set of values $c \in D$.
Due to the \JD and \JEPD conditions, it turns out that feature assertions are actually equivalent to \emph{additional singleton predicates}~$=_c$ that are not part of~\cDom, but can be used in concepts with the same semantics as defined above.%
\footnote{The results in this section also hold for \ALCD, since they can be shown without $\suc$-restrictions, nominals, or restricting to finitely-branching interpretations~\cite{DBLP:conf/dlog/BorgwardtBK24}.}

\begin{restatable}{lemma}{LemFeatureAssertions}\label{lem:feature-assertions}
  For \ALCOSCCD with an $\omega$-admissible concrete domain~$\cDom$ with constants, the following problems are reducible to each other:
  (a)~consistency with additional singleton predicates, and (b)~consistency with feature assertions.
  The reductions take exponential time, but produce ontologies of polynomial size.
\end{restatable}

Additionally, feature assertions can be expressed by \emph{predicate assertions} if \cDom is \emph{homogeneous}, i.e.\ such that every isomorphism between finite substructures of~$\cDom$ can be extended to an isomorphism from~$\cDom$ to itself~\cite{BaRy22}.
All known $\omega$-admissible concrete domains are homogeneous~\cite{BaRy22}.

\begin{restatable}{lemma}{LemHomogeneous}\label{lem:homogeneous}
  For \ALCOSCCD with an $\omega$-admissible and homogeneous concrete domain~\cDom with constants, consistency with feature assertions can be reduced to consistency without feature assertions in exponential time. The resulting ontology is of polynomial size.
\end{restatable}

Together, \Crefrange{lem:feature-assertions}{lem:homogeneous} show that, under these conditions, we can freely use constant values (either in feature assertions or additional singleton predicates) in \ALCOSCCD, without increasing the complexity of reasoning.
The following result then follows together with \Cref{thm:decidability}.

\begin{theorem}\label{thm:homogeneous}
  If \cDom is an \expadm and homogeneous concrete domain with constants, then consistency in \ALCOSCCD with feature assertions and additional singleton predicates is \ExpTime-complete.
\end{theorem}

\section{Undecidable Extensions}
\label{sec:undecidability}
To conclude our investigations, we show that several extensions of \ALCOSCCD, inspired by existing DLs or obtained by seemingly harmless tweaks to the syntax and semantics, are undecidable.
Hereafter, we assume that the domain set of \cDom is infinite and that \cDom is \JD (cf.\ \Cref{sec:preliminaries}).
In this case, we allow w.l.o.g. the usage of set terms $\eqterm{f}{\nx g}$ expressing equality of the values assigned to $f$ and $\nx g$ (we detail the construction of this term \inAppendix).

\paragraph{Comparing set cardinalities and feature values.}
If \cDom is a numerical concrete domain where $D$ is either \Nbb, \Zbb or \Qbb, it is natural to consider comparisons between feature values of an individual $d$ and the cardinalities of sets of role successors of $d$.
For example, we could describe individuals whose age is twice the number of their children using the concept $\suc( \mathsf{age} = 2 \cdot |\mathsf{child}|)$.
This could be achieved by allowing $\suc$-restrictions to contain \emph{mixed numerical constraints} $f = \ell$, where $\ell$ is a PA expression (cf.~\Cref{sec:preliminaries}) that is allowed in \ALCOSCCD and $f \in \NF$; then, we extend $\cdot^\Imc$ by defining $d \in \suc(f = \ell)^{\Imc}$ iff $f^{\Imc}(d) = \sigma_d(\ell)$.
Unfortunately, for the CDs considered here, this leads to undecidability, which can be shown by a reduction to
$\ALC(\cDom)$ with the concrete domain $\cDom=(\Nbb,+_1)$ where $+_1$ is the successor relation, which is known to be undecidable~\cite{BaRy22}.
\begin{theorem}
    If \cDom is a numerical concrete domain that is \JD, then consistency of \ALCOSCCD TBoxes with mixed numerical constraints is undecidable.
\end{theorem}
\begin{proof}
    We force $r \in \NR$ to be functional with the CI $\top \sqsubseteq \suc(|r| \le 1)$.
    We encode the CD-restriction $C := \exists p_0, p_1.{+_1}$ using $C_0\sqcap C_1$, where
    \begin{equation*}
        C_i := \begin{cases}
            \suc(f_i = |S| + i) & \text{if $p_i = f_i$} \\ \suc(f_i' = |S| + i) \sqcap \suc(|r_i \cap (f_i' = \nx f_i)| \ge 1) & \text{if $p_i = r_i f_i$}.
        \end{cases}
    \end{equation*}
    for $i = 0,1$, with fresh names $S \in \NC$, $f_i' \in \NF$.
    \qed
\end{proof}

\paragraph{Local and global cardinality constraints.}
It is possible to extend \ALCSCC by replacing $\suc$-restrictions, ranging over sets of role successors, with \emph{$\sat$-restrictions} $\sat(\con)$ ranging over the whole domain of an interpretation.
For the resulting DL, called \ALCSCCpp, the consistency problem is \NExpTime-complete~\cite{BaBeRu20}.
In this DL, we can state that an individual likes \emph{all} existing cars using the concept $\sat(\mathsf{likes} \cap \mathsf{Car} = \mathsf{Car})$; in contrast, $\suc(\mathsf{likes} \cap \mathsf{Car} = \mathsf{Car})$ describes an individual that likes all cars to which it is related by some role.

If we consider the DL $\ALCSCCpp(\cDom)$ obtained by adding $\sat$-restrictions in the presence of concrete domains, then these restrictions may additionally contain feature roles.
For example, the concept $\sat(\top = (\mathsf{age} \ge \nx \mathsf{age}))$ describes the \emph{overall} oldest individuals, by saying that their age is greater or equal to those of all individuals, while $\suc(\top = (\mathsf{age} \ge \nx \mathsf{age}))$ describes individuals that are not younger than any individuals related to them by some role name.

Formally, both \ALCSCCpp and $\ALCSCCpp(\cDom)$ are evaluated over \emph{finite} interpretations.
In~\cite{BaBeRu20}, this has been used to show that the consistency problem for the extension of \ALCSCCpp with inverse roles is undecidable.
Similarly, we can use $\sat$-restrictions with feature roles to simulate multiplication of cardinalities of \emph{finite} sets, and thus reduce Hilbert's tenth problem~\cite{HilbertTenth} to the consistency of a $\ALCSCCpp(\cDom)$ TBox, provided that \cDom is \JD.
Writing $C \equiv D$ as a shorthand for $C \sqsubseteq D$ and $D \sqsubseteq C$, we can encode the equation $\efr = (x = y \cdot z)$ over integers as a product of cardinalities $|A_x^{\Imc}| = |A_y^{\Imc}| \cdot |A_z^{\Imc}|$, in three steps.
First, we enforce $r_{\efr}^{\Imc} = A_y^{\Imc} \times A_z^{\Imc}$ to hold with $A_y \equiv \sat(|r_{\efr}| \ge 1)$ and $A_y \equiv \sat(r_{\efr} = A_z)$; then, we enforce $|s_{\efr}^{\Imc}| = |A_x^{\Imc}|$ by adding $\top \sqsubseteq \sat(s_{\efr} = \eqterm{f_{\efr}}{\nx g_{\efr}})$ and the CIs 
\begin{equation*}
    \top \sqsubseteq \sat(|\eqterm{g_{\efr}}{\nx f_{\efr}}| \le 1) \;\text{and}\; A_x \sqsubseteq \sat(|\eqterm{g_{\efr}}{\nx f_{\efr}}| \ge 1).
\end{equation*}
Finally, we add $\top \sqsubseteq \sat(|r_{\efr}| = |s_{\efr}|)$, so that, for every finite model \Imc of all these CIs, $|A_x^{\Imc}| = |s_\efr^{\Imc}| = |r_\efr^{\Imc}| = |A_y^{\Imc} \times A_z^{\Imc}| = |A_y^{\Imc}| \cdot |A_z^{\Imc}|$ holds.
\begin{theorem}
    If the concrete domain \cDom is infinite and \JD, then the consistency problem for \ALCSCCppD TBoxes is undecidable.
\end{theorem}

\paragraph{Transitive roles.}
Often, we may want a role name to be interpreted as a transitive relation: for instance, $\mathsf{trans}(\mathsf{ancestor})$ in the TBox expresses the fact that the ancestor of an ancestor is also an ancestor.
The interaction between number restrictions and transitivity axioms in the presence of role inclusions is known to lead to undecidability~\cite{HoSaTo00}.
It is possible to regain decidability by disallowing transitive roles within number restrictions, even in the presence of inverse roles~\cite{HoSaTo00}.
Another restriction that leads to decidability is to replace number restrictions with role functionality axioms; in this case, decidability holds even if one additionally allows nominals and inverse roles~\cite{DBLP:conf/aaai/Gutierrez-Basulto17}.

In the DL \SSCC that extends \ALCSCC with transitivity axioms, consistency is undecidable even under all syntactic constraints mentioned above.
In particular, we require that numerical constraints contain no transitive roles and no constants other than~$0$ or~$1$.
By adapting the reduction~\cite{HoSaTo00} from the tiling problem, which is known to be undecidable~\cite{Be66}, we obtain the following.

\begin{theorem}
    \label{thm:undecidability-transitive}
    Consistency in \SSCC is undecidable, even if numerical constraints contain no transitive roles and no constants other than~$0$ or~$1$.
\end{theorem}

\section{Conclusion}
\label{sec:conclusion}

We have presented the very expressive DL \ALCOSCCD that supports concrete domain restrictions and role successor restrictions involving feature values.
We have shown that consistency in this logic is \ExpTime-complete, the same as for the basic DL \ALC.
Moreover, we have discussed the consequences of adding assertions, transitive roles, unrestricted semantics, or mixed constraints, most of which make the logic undecidable.
While feature roles can already express a restricted form of inverse roles, in the future, we would like to investigate the decidability and complexity of $\mathcal{ALCOISCC}(\cDom)$ with full inverse roles, for which it is known that they increase the complexity of classical DLs with nominals and number restrictions to \NExpTime~\cite{DBLP:journals/jair/Tobies00}.
Another avenue of research is to implement a reasoner for \ALCOSCCD, based on a suitable tableaux algorithm~\cite{LuMi07} that needs to integrate a QFBAPA solver and a concrete domain reasoner.
Currently, reasoners for DLs with non-trivial concrete domains only exist for \ALCD and \ELD with so-called $p$-admissible concrete domains and without feature paths~\cite{DBLP:conf/rulemlrr/AlrabbaaBBKK23}.

\begin{credits}
\subsubsection{\ackname} This work was partially supported by DFG grant 389792660 as part of TRR~248 -- CPEC, by the German Federal Ministry of Education and Research (BMBF, SCADS22B) and the Saxon State Ministry for Science, Culture and Tourism (SMWK) by funding the competence center for Big Data and AI \enquote{ScaDS.AI Dresden/Leipzig}.
\subsubsection{\discintname}
The authors have no competing interests to declare that are
relevant to the content of this article.
\end{credits}

%
\appendix
\section{Deciding Consistency --- Auxiliary Results}
\label{app:satisfiability}

The proofs of \Cref{lem:soundness} and~\Cref{lem:completeness} depend on a series of lemmas that we provide in this section.

\subsection{Soundness of~\Cref{alg:type-elimination}}

We observe that for every patchwork \cDom, the fact that \cDom satisfies \AP (see~\Cref{sec:preliminaries}) implies the following property:
\begin{description}
    \item[$\mathsf{AP}^{+}$]\label{prop:ap-plus} if $\Cbb$ is a finite set of constraint systems and $V$ a set of variables such that $V(\Bfr) \cap V(\Cfr) = V$ and \Bfr and \Cfr agree on $V$ for all $\Bfr, \Cfr \in \Cbb$, then the constraint system $\bigcup \Cbb$ is satisfiable iff each $\Cfr \in \Cbb$ is satisfiable.
\end{description}

We can now prove the claims made in the proof of \Cref{lem:soundness}.
\begin{lemma}
   \label{lem:soundness-features}
   For every $m \in \Nbb$, the constraint system $\cs^m$ has a solution.
\end{lemma}
\begin{proof}
    Since \cDom is in particular a patchwork, we can extend every constraint system $\cs$ over \cDom with a solution $h$ to a complete and satisfiable constraint system.
    Indeed, the fact that \cDom is \JEPD implies that, for all $k \in \Nbb$, either \cDom has no $k$-ary relations or for all variables $v_1$, \ldots, $v_k$ the system $\cs$ contains at most one constraint $P(v_1,\dotsc,v_k)$ with a $k$-ary relation~$P$ of~\cDom.
    If $\cs$ contains no such constraint, then we can add to it the unique constraint $P(v_1,\dotsc,v_k)$ for which $(h(v_1),\dotsc,h(v_k)) \in P^D$ holds.
    In this case, $h$ is still a solution of the extended system.
    Hereafter, we call a \emph{completion} of a satisfiable constraint system $\cs$ a complete constraint system obtained using this procedure.

    Our second observation is that for all $(\ain,w) \in \Delta^{\Imc}$ the complete constraint system $\cs_{\ain,w}$ is satisfiable.
    Assuming that $\wend(\ain, w) = (t, V, \cs_{t, V}) \in \Tbb$, we know from~\Cref{dfn:augmented-type} that ${\cs_{t, V}}$ has a solution $h'$.
    We extend $h'$ to a solution $h$ of $\cs_{\ain,w}$ by first setting $h'(f^{(v,i)}) := h'(f^{(v,V(v))})$ for $i = V(v) + 1,\dotsc,\sigma_{\ain,w}(|v|)$ if $f^{(v,V(v))}$ occurs in $\cs_{t, V}$ and then renaming the variables in the domain of $h'$ using the renaming used to construct $\cs_{\ain,w}$, thus obtaining $h$.

    We prove the claim by induction over $m \in \Nbb$. For $m = 0$, the system $\cs^0$ corresponds to the union of all systems $\cs_{\ain,\varepsilon}$ with $\ain \in \Ibb$.
    Every two systems of this form share exactly the variables $f^{\mathfrak{b},\varepsilon}$ for which $f^{\mathfrak{b}}$ occurs in $\csi$; moreover, they agree on these variables, since they all include (up to renaming) a copy of the complete constraint system $\csi$.
    By $\mathsf{AP}^{+}$, we conclude that $\cs^0$ is satisfiable.

    Now, we inductively assume that the constraint system $\cs^m$ has a solution for $m \in \Nbb$ and show that this implies that $\cs^{m+1}$ is satisfiable, too.
    We notice that $\cs^{m+1}$ is the union of $\cs^m$ and all systems $\cs_{\ain,w}$ with $(\ain,w) \in \Delta^{m+1} \setminus \Delta^m$.
    By construction, for every $(\ain,w) \in \Delta^{m+1} \setminus \Delta^m$ there is a unique $(\ain, w') \in \Delta^m$ such that $w = w' \cdot (\agt,v,i)$ for some augmented type $\agt \in \Tbb$, some Venn region $v$ of \Tmc and some natural number $i$.
    In particular, $\agt$ patches $\wend(\ain,w')$ at $(v,i)$ by definition of $\Delta^{m+1}$.
    Let $\cs'$ be a completion of $\cs^m$ and consider the complete constraint system $\cs_{\ain,w}$.
    The variables shared by $\cs'$ and $\cs_{\ain,w}$ are exactly (up to renaming) those occurring in $\csi$ and those of the form $f^{\ain,w}$ for $f \in \NF$ that occur in the complete constraint systems $\cs_{\ain,w'}$ and $\cs_{\ain,w}$.
    Due to the fact that $\agt$ patches $\wend(\ain,w')$ at $(v,i)$, we deduce that these two systems agree on their shared variables.
    Given that $\cs_{\ain,w'}$ is a subsystem of $\cs'$ and that both $\cs'$ and $\cs_{\ain,w}$ are complete, we deduce that these two systems agree on their shared variables, and since both are satisfiable we conclude that $\cs' \cup \cs^w_{\ain}$ is satisfiable.
    By taking a completion of this constraint system and iteratively repeating this process for every other element in $(\ain,w) \in \Delta^{m+1} \setminus \Delta^m$, we obtain a satisfiable constraint system that includes $\cs^{m+1}$ as a subsystem, which is then satisfiable.
    \qed
\end{proof}

\begin{lemma}
    \label{lem:soundness-concepts}
    For
    all
    $d = (\ain,w) \in \Delta^{\Imc}$ and $C \in \Mmc$, $C \in \agroot(\wend(d))$ iff $d \in C^{\Imc}$.
\end{lemma}
\begin{proof}
    We prove this claim by structural induction over $C \in \Mmc$.
    We assume that $\wend(\ain,w) = \agt = (t, V, \cs_{t,V})$ and first prove the base cases where $C$ is either a concept name, a nominal or an existential CD-restriction.
    \begin{itemize}
        \item The case $C = A \in \NC$ is trivially covered by the definition of $A^{\Imc}$.
        \item If $C = \{a\}$, we notice that $\{a\} \in \agroot(\agt)$ iff $\agt = \agt_{\ain}$ with $a \in \ain$ and $\ain_{\agt} = \ain$; by construction of $\Delta^{\Imc}$, this holds iff $(\ain,w) = (\ain,\varepsilon) \in \{a\}^{\Imc}$.
        \item Let $C = \exists p_1,\dotsc,p_k. P \in \Mmc$. If $C \in \agroot(\agt)$, then~\Cref{dfn:local-system} together with $\agt \in \Tbb$ implies that there is a constraint $P(f_1^{x_1},\dotsc,f_k^{x_k}) \in \cs_{t,V}$ such that for $i = 1,\dotsc,k$ either $x_i = \ain_i$ with $\ain_i \in \Ibb$ or $x_i = \star$ or $x_i = (v_i,j_i)$ for some $v_i \in V$ with $1 \le j_i \le V(v_i)$.
        In the last case, we know that there is $\agt_{(v_i,j_i)} \in \Tbb$ that patches $\agt$ at $(v_i,j_i)$.
        Using these indices and augmented types, we select for $i = 1,\dotsc,k$ the domain elements
        \begin{equation*}
        d_i = (\ain_i, w^i) := 
        \begin{cases}
            (\ain_i, \varepsilon) & \text{if $x_i = \ain_i \in \Ibb$} \\
            (\ain, w) & \text{if $x_i = \star$} \\
            (\ain, w \cdot (\agt_{(v_i,j_i)}, v_i, j_i)) & \text{if $x_i = (v_i,j_i)$}.
        \end{cases}
        \end{equation*}
        Then, $P(f_1^{\ain_1,w^1},\dotsc,f_k^{\ain_k,w^k}) \in \cs^{\Imc}$ holds and thus $(f_1^{\Imc}(d_1),\dotsc,f_k^{\Imc}(d_k)) \in P^D$ by definition of \Imc.
        We prove that $(\ain,w) \in C^{\Imc}$ by showing that $f_i^{\Imc}(d_i) \in p_i^{\Imc}(\ain,w)$ holds for $i=1,\dotsc,k$:
        \begin{itemize}
            \item If $p_i=f_i$, then $x_i=\iota(t)$ is either $\star$ if $t$ is anonymous or $x_i=\ain_i=\ain_t=\ain$ and $w^i=w=\varepsilon$.
            In the former case, $f_i^\Imc(d_i) = f_i^\Imc(\ain,w) \in p_i^\Imc(\ain,w)$.
            In the latter case, we have $f_i^\Imc(d_i) = f_i^\Imc(\ain_i,\varepsilon) \in p_i^\Imc(\ain,w)$.
            \item If $p_i=r_if_i$, then $x_i=\iota((v_i,j_i))$ is either $(v_j,j_i)$ if $v_j$ is anonymous or $x_i=\ain_i=\ain_{v_i}$ and $w^i=\varepsilon$.
            Moreover, in both cases, we have $X_{r_i} \in v_i$, and thus $((\ain,w),(\ain,w^i)) \in r_i^\Imc$ or $((\ain,w),(\ain_{v_i},\varepsilon))\in r_i^\Imc$, respectively,
            Therefore, $f_i^\Imc(d_i) \in p_i^\Imc(\ain,w)$ holds in both cases.
        \end{itemize}

        Vice versa, assume that $C \notin \agroot(\agt)$.
        We show that for all $(c_1,\dotsc,c_k)$ where $c_i \in p_i^{\Imc}(w)$ for $i = 1,\dotsc,k$ it holds that $(c_1,\dotsc,c_k) \notin P^D$.
        We choose individuals $d_i = (\ain_i,w^i)$ for $i = 1,\dotsc,k$ such that $d_i = (\ain,w)$ and $f_i^\Imc(d_i) = c_i$ if $p_i = f_i$  and $d_i = (\ain',w')$ is an $r_i$-successor of $(\ain, w)$ with $f_i^{\Imc}(d_i) = c_i$ if $p_i = r_i f_i$.
        The construction of $\cs_{\ain,w}$ from $\cs_{t,V}$, together with $C \notin \agroot(\agt)$ and thus $\neg C \in \agroot(\agt)$, implies that $P(f_1^{\ain_1,w^1},\dotsc,f_k^{\ain_k,w^k}) \notin \cs_{\ain,w}$ according to~\Cref{dfn:local-system}.
        Since $f^{\Imc}(d_i)$ is defined, however, $f_i^{\ain_i,w^i}$ must occur in $\cs_{\ain,w}$ for $i = 1,\dotsc,k$.
        Given that $\cs_{\ain,w}$ is complete, there must be $P'(f_1^{\ain_1,w^1},\dotsc,f_k^{\ain_k,w^k}) \in \cs_{\ain,w}$ for some $P' \ne P$ $k$-ary.
        The interpretation of features of \Imc is a solution of $\cs^{\Imc}$ and thus of $\cs_{\ain,w}$, so we deduce that $(c_1,\dotsc,c_k) = (f_1^{\Imc}(d_1),\dotsc,f_k^{\Imc}(d_k)) \in (P')^D$, hence $(c_1,\dotsc,c_k) \notin P^D$ (by \JEPD).
    \end{itemize}

    Assume that the claim holds for $C_1, C_2, D \in \Mmc$ to prove the inductive cases.
    \begin{itemize}
        \item If $C = \neg D \in \Mmc$, then $C \in \agroot(\agt)$ iff $D \notin \agroot(\agt)$ iff $(\ain, w) \notin D^{\Imc}$ iff $(\ain, w) \in C^{\Imc}$, where the equivalences hold due to~\Cref{dfn:type}, the inductive hypothesis, and the semantics of negation, respectively.
        \item If $C = C_1 \sqcap C_2 \in \Mmc$, then $C \in \agroot(\agt)$ iff $C_i \in \agroot(\agt)$ for $i = 1,2$ iff $(\ain,w) \in C_i^{\Imc}$ for $i = 1,2$ iff $(\ain,w) \in C^{\Imc}$, with equivalences justified as before.
        \item If $C = \suc(\con) \in \Mmc$, we show that the solution $\sigma_{\ain,w}$ of $\phi_{t,V}$ used to determine the role successors of $(\ain,w)$ satisfies the formula $\phi_{\con}$ iff the QFBAPA assignment $\zeta_{\ain,w}$ induced by $(\ain,w)$ in~\Imc is a solution of $\con$ (recall that $\zeta_{\ain,w}(\Umc)=\ars^\Imc(\ain,w)$).
        Since $C \in \agroot(\agt)$ iff $\phi_{\con}$ occurs in $\phi_{t,V}$, this allows us to conclude that $C \in \agroot(\agt)$ iff $\zeta_{\ain,w}$ is a solution of $\con$, which, by the semantics of \ALCOSCCD, happens iff $(\ain,w) \in C^{\Imc}$.

        To show that $\sigma_{\ain,w}$ satisfies $\phi_\con$ iff $\zeta_{\ain,w}$ satisfies $\con$, we construct a bijection $\pi \colon \zeta_{\ain,w}(\Umc) \to \sigma_{\ain,w}(\Umc)$ such that $(\ain',w') \in \zeta_{\ain,w}(\delta)$ iff $\pi(\ain',w') \in \sigma_{\ain,w}(X_\delta)$ for each role name, concept or feature role~$\delta$.
        For every named element $(\ain',\varepsilon)\in\zeta_{\ain,w}(\Umc)$, we define $\pi(\ain',\varepsilon)$ to be the unique element in $\sigma_{\ain,w}(\bigcap_{a\in\ain'} X_{\{a\}})$: since $(\ain',\varepsilon)$ is a role successor of $(\ain,w)$, we know that $\agt_{\ain'}$ patches $\agt$, which means that a Venn region $v$ with $S_v\subseteq\agroot(\agt_{\ain'})$ and thus $\ain_v=\ain'$ must occur in~$\supp(V)$, which means that $\sigma_{\ain,w}(\Umc)$ contains at least one element satisfying $\bigcap_{a\in\ain'} X_{\{a\}}$; moreover, there can be at most one such element since $\sigma_{\ain,w}$ satisfies $\phi_t$.
        For each anonymous element $(\ain,w')$ with $w'=w\cdot(\agt_{(v,i)},v,j)$, we can assign a unique element $\pi(\ain',w)\in\sigma_{\ain,w}(v)$, since we created exactly $\sigma_{\ain,w}(|v|)$ many such elements~$w'$.
        This mapping~$\pi$ is a bijection since named Venn regions must have cardinality $1$ w.r.t.~$\sigma_{\ain,w}$, we explicitly created enough elements for each anonymous Venn region $v\in\supp(V)$, each such element is a role successor of $(\ain,w)$ due to the constraint $X_{r_1}\cup\dots\cup X_{r_n} = \Umc$, and there are no other role successors of $(\ain,w)$ in~\Imc.

        It remains to show that $(\ain',w') \in \zeta_{\ain,w}(\delta)$ iff $\pi(\ain',w') \in \sigma_{\ain,w}(X_\delta)$.
        For anonymous role successors $(\ain, w')$ with $w' = w \cdot (\agt_{(v,i)}, v, j)$ of $(\ain, w)$, this amounts to showing the following:
        \begin{itemize}
            \item $X_r \in v$ iff $(\ain,w') \in r^{\Imc}((\ain,w))$ for all role names~$r\in\NR$;
            \item $X_D \in v$ iff $(\ain,w') \in D^{\Imc}$ for all concepts~$D\in\Mmc$; and 
            \item $X_\gamma \in v$ iff $(\ain,w') \in \gamma^{\Imc}((\ain,w))$ for all feature roles~$\gamma$.
        \end{itemize}
        The first point is a consequence of the definition of \Imc.
        Using again this definition, we observe that $\agt'$ patches $\agt$ at $(v,i)$, and thus $S_v \subseteq \agroot(\agt_{(v,i)})$ holds.
        Since for $D \in \Mmc$ we have that $D \in S_v$ iff $X_D \in v$ and by inductive hypothesis we have that $D \in \agroot(\agt_{(v,i)})$ iff $(\ain,w') \in D^{\Imc}$, we conclude that $(\ain,w') \in D^{\Imc}$ iff $X_D \in v$.
        Finally, by~\Cref{dfn:local-system} and using the renaming we adopted before, for $\gamma := P(\alpha_1,\dotsc,\alpha_k)$ we have that $X_\gamma \in v$ iff $P(f_1^{x_1},\dotsc,f_k^{x_k}) \in \cs_{\ain,w}$ where $x_i = (\ain, w)$ if $\alpha_i = f_i$ and $x_i = (\ain,w')$ if $\alpha_i = \nx f_i$.
        By construction of \Imc and the definition of $x_i$, we have that $P(f_1^{x_1},\dotsc,f_k^{x_k}) \in \cs^{\Imc}$ iff $(f_1^{\Imc}(x_1),\dotsc,f_k^{\Imc}(x_k)) \in P^D$, and by the semantics of \ALCOSCCD this happens iff $(\alpha_1^{\Imc}((\ain,w),(\ain,w')),\dotsc,\alpha_k^{\Imc}((\ain,w),(\ain,w'))) \in P^D$, i.e. iff $(\ain,w') \in \gamma^{\Imc}(\ain,w)$.

        We can show that a similar characterization for named role successors $(\ain',\varepsilon)$ of $(\ain,w)$, by considering the unique singleton Venn region~$v$ of $\sigma_{\ain,w}$ that satisfies $\bigcap_{a\in\ain'} X_{\{a\}}$.
        \qed
    \end{itemize}
\end{proof}

\subsection{Completeness of~\Cref{alg:type-elimination}}

\begin{lemma}\label{claim:completeness-bag}
    There is a Venn bag $V_d$ for $t_{\Imc}(d)$ w.r.t.~\Tmc such that $V_d(v_e) = n_e$ for all $e \in \wit$, and for all other $v\in\supp(V_d)$ we have $V_d(v)=1$ and there is a role successor $e\in\Delta^\Imc\setminus\wit$ of~$d$ such that $v=v_e$.
\end{lemma}
\begin{proof}
    Given that $\wit$ contains at most \mt elements, there can be at most \mt different constraints $|v_e| \ge n_e$ generated by elements $e$ of $\wit$ and $1 \le n_e \le \mt + 1$ holds for all these constraints.
    Consider the QFBAPA formula $\phi_{d}$ that is the conjunction of $\phi_{t_\Imc(d)}$ and the constraints $|v_e|\ge n_e$ for $e\in\wit$.
    By construction, the QFBAPA assignment $\sigma_d$ induced by \Imc (see \Cref{sec:dls}) is a solution of $\phi_{t_\Imc(d)}$ with $\sigma_d(\mathcal{U}) = \ars^{\Imc}(d)$, which satisfies the additional constraints $|v_e|\ge n_e$ with $e \in \wit$ since $\wit\subseteq \ars^{\Imc}(d)$, and is thus a solution of $\phi_d$.
    By \Cref{lem:preliminaries:savior}, there is another solution $\sigma_d'$ of $\phi_d$ for which there are at most \nt Venn regions $v$ with $\sigma_d'(v)\neq\emptyset$, and, whenever $\sigma_d'(v)\neq\emptyset$, then also $\sigma_d(v)\neq\emptyset$; that is, every non-empty Venn region~$v$ according to $\sigma_d'$ still corresponds to an element $e\in \ars^{\Imc}(d)$ such that $v_e=v$.
    We define the bag $V_d$ as follows:
    \[
      V_d(v) := \begin{cases}
        n_e & \text{if $v = v_e$ and $e \in \wit$,} \\
        1 & \text{if $v \ne v_e$ for $e \in \wit$ and $\sigma_d'(v)\neq\emptyset$, and} \\
        0 & \text{otherwise.}
      \end{cases}
    \]
    Then $\supp(V_d)$ contains at most \mt Venn regions $v$, all satisfying $V_d(v) \le \mt + 1$.
    Since $\sigma_d'$ is also a solution of~$\phi_{t_\Imc(d),V_d}$ and $\supp(V_d)$ contains exactly the Venn regions that are assigned non-empty sets by~$\sigma_d'$, we conclude that $V_d$ satisfies all conditions of~\Cref{dfn:venn-bag}, hence that $V_d$ is a Venn bag for $t_\Imc(d)$.
    \qed
\end{proof}

\begin{lemma}\label{claim:completeness-local-system}
   $\cs_d$ is a satisfiable local system for $t_{\Imc}(d)$ and $V_d$.
\end{lemma}
\begin{proof}
    Clearly, $h_d(f^{x}) := f^{\Imc}(\xi_d(x))$ for $f\in\NF$, $x \in X_d$ defines a solution of~$\cs_d$.
    We show that $\cs_d$ is a local system for~$t_{\Imc}(d)$ and~$V_d$, according to~\Cref{dfn:local-system}:
  \begin{enumerate}
    \item We have $\exists p_1,\dots,p_k.P\in t_\Imc(d)$ iff $(c_1,\dots,c_k)\in P^D$ for some values $c_i\in p_i^\Imc(d)$ with $i = 1,\dotsc,k$; by construction of~$V_d$, if $p_i=r_if_i$ we can find $e_i\in \wit$ such that $(d,e_i) \in r_i^\Imc$ and $f_i^\Imc(e_i)=c_i$, and set $x_i := \lambda_d^{-1}(e_i)$; if $p_i = f_i$ (and $f_i^\Imc(d)=c_i$), we set $x_i := \iota(t_\Imc(d))$.
    In both cases, we have $f_i^\Imc(\xi_d(x_i))=c_i$.
    By definition of $\cs_d$, then, $(c_1,\dots,c_k)\in P^D$ iff $P(f_1^{x_1},\dotsc,f_k^{x_k}) \in \cs_d$.

    \item Consider any variable of the form $X_{P(\alpha_1,\dots,\alpha_k)}$, a Venn region $v\in\supp(V_d)$ and $1\le j\le V_d(v)$. We know that $v = v_e$ for some individual $e=\lambda_d(x)$ with $x \in X_d$, and thus $X_{P(\alpha_1,\dots,\alpha_k)} \in v = v_e$ iff $\big(\alpha_1^\Imc(d,e),\dots,\alpha_k^\Imc(d,e)\big) \in P^D$.
    Setting $x_i := \iota(t_\Imc(d))$ if $\alpha_i = f_i$ and $x_i := x$ if $\alpha_i = \nx f_i$, we obtain that $\big(\alpha_1^\Imc(d,e),\dots,\alpha_k^\Imc(d,e)\big) \in P^D$ iff $\big(f_1^\Imc(\xi_d(x_1)),\dots,f_k^\Imc(\xi_d(x_k))\big) \in P^D$ iff $P(f_1^{x_1},\dots,f_k^{x_k}) \in \cs_d$.
    \qed
  \end{enumerate}
\end{proof}

\begin{lemma}\label{claim:completeness-patching}
   Every augmented type in $\Tbb_\Imc$ is patched in $\Tbb_\Imc$.
\end{lemma}
\begin{proof}
    Consider an augmented type $\agt_\Imc(d) = \big( t_\Imc(d), V_d, \cs_d \big) \in \Tbb_\Imc$ as defined above, $v\in\supp(V_d)$, and $1\le i\le V_d(v)$.
    We consider $e:=\xi_d((v,i))$ if $v$ is anonymous, and $e:=\ain_v^\Imc=\xi_d(\ain_v)$ otherwise.
    In both cases, we have $v_e=v$ by construction of~$\xi_d$.
    We show that $\agt_{\Imc}(e)$ patches $\agt_\Imc(d)$ at $(v,i)$.
    First,
    \begin{align*}
      S_v
      & = \{ C \in \Mmc \mid X_C \in v \} \cup \{ \lnot C \in \Mmc \mid X_C^c \in v \} \\
      & \subseteq \{ C \in \Mmc \mid C \in t_\Imc(e) \} \cup \{ \lnot C \in \Mmc \mid C \notin t_{\Imc}(e) \} \\
      & = t_\Imc(e).
    \end{align*}
    Second, we consider the system $\cs_d \merge{(v,i)} \cs_e$ obtained by renaming the variables in $\cs_e$ as in \Cref{dfn:augmented-type}.
    As discussed above, $h_d(f^x) := f^{\Imc}(\xi_d(x))$ defines a solution of~$\cs_d$, and similarly $h_e(f^x) := f^{\Imc}(\xi_e(x))$ is a solution of~$\cs_e$.
    In particular, $h_d(f^{\ain}) = f^{\Imc}(\ain^\Imc) = h_e(f^{\ain})$ for all $\ain\in\Ibb$ and, if $e$ is anonymous, $h_d(f^{(v,i)}) = f^\Imc(\xi_d((v,i))) = f^\Imc(e) = f^\Imc(\xi_e(\star)) = h_e(f^\star)$.
    Thus, the mapping $h$ defined as the union of $h_d$ and $h_e$ (after renaming) is a solution of $\cs_d \merge{(v,i)} \cs_e$.
    \qed
\end{proof}

\subsection{Complexity}

\ThmDecidability*
\begin{proof}
  \ExpTime-hardness follows from \ExpTime-hardness for \ALC~\cite{DBLP:conf/ijcai/Schild91}.
  It remains to show that \Cref{alg:type-elimination} can be executed in deterministic exponential time by enumerating all choices in Lines~1 and~2 instead of guessing them.
  First, there are only exponentially many individual type systems \Ibb and individual constraint systems for Line~1 of \Cref{alg:type-elimination}.
  Moreover, there are only exponentially many augmented types $(t,V,\cs_{t,V})$ since the size of Venn bags~$V$ is bounded polynomially and thus also $\cs_{t,V}$ can contain only polynomially many variables and constraints.
  In addition, satisfiability of the polynomially large $\phi_{t,V}$ and $\cs_{t,V}$ can be checked in exponential time since satisfiability of QFBAPA formulae is NP-complete~\cite{KuRi07} and $\cDom$ is \expadm, respectively.
  Therefore, the initial set \Tbb in Line~3 can be constructed in exponential time and there are also only exponentially many possibilities to assign augmented types~$\agt_\ain$ to the individual types in~\Ibb in Line~2.
  Since each iteration of the loop in Line~4 removes one augmented type from \Tbb, there can be at most exponentially many iterations.
  Each iteration can be performed in exponential time, as each check for patching involves a polynomial test to check whether $S_v \subseteq t'$ and an exponential check for satisfiability of a constraint system of polynomial size, and at most exponentially many patching checks occur.
  By \Cref{lem:soundness,lem:completeness}, we conclude that consistency of an \ALCOSCCD TBox is decidable in exponential time.
  \qed
\end{proof}

\section{Reasoning with ABoxes --- Proofs}
\label{app:constants}

\paragraph{Referring to feature values of named individuals}
We can use the roles~$\mathsf{ref}_a$ (mentioned in \Cref{sec:dls,sec:constants}) in arbitrary CD-restrictions to refer to the feature values of named individuals.
For example, $(\exists \mathsf{child}\, \mathsf{salary}, \mathsf{ref}_\mathtt{Sam} \,\mathsf{salary}.{<})$ compares the salary of an anonymous child to that of \texttt{Sam}.

We introduce a related construction here, variants of which will be used in several of the following proofs.
The idea is to introduce features like $\mathsf{salary}_{\mathtt{Sam}}$ that can be used to express feature roles like $\nx \mathsf{salary} < \mathsf{salary}_{\mathtt{Sam}}$ within $\suc$-restrictions, in order to quantify the number of successors with a $\ex{salary}$ smaller than \texttt{Sam}'s.
For this purpose, the interpretation of the feature~$f_a$ needs to be such that $f_a^\Imc(d)$ is equal to $f^\Imc(a^\Imc)$ at every individual $d\in\Delta^\Imc$.
The idea is to use the role $\mathsf{ref}_a$ to enforce this, using a CI like $\top\sqsubseteq\forall\mathsf{ref}_a f,f_a.{=}$.
However, \cDom does not necessarily contain the equality predicate~$=$, which means that this may not be a valid CI in \ALCOSCCD.
Nevertheless, by \JD, we know that there is a quantifier-free, equality-free first-order formula $\phi_=(x,y)$ over the signature of \cDom that is equivalent to \mbox{$x=y$}.
Moreover, by \JEPD and finiteness of the signature, we can express negated atoms as disjunctions of positive atoms, so that we may assume $\phi_=(x,y)$ to be a disjunction of conjunctions of positive atoms.

We can use this to construct the concept $C_{\mathsf{ref}_a f=f_a}$ that is obtained from $\phi_=(x,y)$ by replacing $\land$ with $\sqcap$, $\lor$ with $\sqcup$ and every atom $P(t_1,\dots,t_n)$ with $\forall p_1,\dots,p_n.P$, where $p_i=\mathsf{ref}_a f$ whenever $t_i=x$ and $p_i=f_a$ whenever $t_i=y$.
This concept is equivalent to the intended CD-restriction $\forall \mathsf{ref}_a f, f_a.{=}$ since $\mathsf{ref}_a$ is functional, i.e.\ every individual has exactly one $\mathsf{ref}_a$-successor, namely~$a$.
Thus, the CI $\top\sqsubseteq C_{\mathsf{ref}_a f=f_a}$ enforces that, whenever both $f^\Imc(a^\Imc)$ and $f_a^\Imc(d)$ are defined, then they must be equal.
Finally, we can complement this CI by similar ones to express that these feature values are either both defined or both undefined: $\exists f_a,f_a.{=}\sqsubseteq\exists \mathsf{ref}_af,\mathsf{ref}_af.{=}$ and $\exists \mathsf{ref}_af,\mathsf{ref}_af.{=}\sqsubseteq\exists f_a,f_a.{=}$
(we can construct concepts $C_{f_a=f_a}$ and $C_{\mathsf{ref}_af=\mathsf{ref}_af}$ equivalent to $\exists f_a,f_a.{=}$ and $\exists \mathsf{ref}_af,\mathsf{ref}_af.{=}$, respectively, similarly to $C_{\mathsf{ref}_af=f_a}$ above).

\LemFeatureAssertions*
\begin{proof}
    We can express every feature assertion $f(a,c)$ by $\{a\}\sqsubseteq\exists f.{=_c}$.
	For the other direction, consider an $\ALCOSCCD$ TBox~\Tmc that uses additional singleton predicates.
    We show how to construct a TBox~$\Tmc'$ and an ABox~$\Amc'$ that simulate all additional singleton predicates~$=_c$ in~\Tmc by using feature assertions.
    Since $=_c$ is unary, it can occur only in CD-restrictions of the form $\exists f.{=_c}$ or $\exists rf.{=_c}$ and feature roles ${=_c}(f)$ or ${=_c}(\nx f)$.
    CD-restrictions $\exists rf.{=_c}$ can be equivalently expressed as $\suc(|r \cap \exists f.{=_c}| \ge 1)$, and ${=_c}(\nx f)$ can directly be replaced by $\exists f.{=_c}$.
    This means that we can assume that $=_c$ occurs only in expressions of the form $\exists f.{=_c}$ or ${=_c}(f)$.
    
    The main idea is to store the value~$c$ in a special feature~$f_c$ by using feature assertions, and make sure that the value of~$f_c$ is equal to~$c$ at every element reachable from a named individual by a role chain.
    We can then express $\exists f.{=_c}$ and ${=_c}(f)$ by making $f$ equal to~$f_c$, for which we exploit~\JD.
	
	First, we ensure that $f_c$ is a total function by adding the axiom $\top\sqsubseteq\exists f_c.\top_\cDom$ to~$\Tmc'$, where $\top_\cDom$ is interpreted as~$D$.
    Although $\top_\cDom$ may not be a predicate of $\cDom$, by \JEPD and the fact that the signature of~$\cDom$ is non-empty and finite, $\top_\cDom$ can be expressed as the disjunction of some $k$-ary predicates $P_1,\dots,P_m$.
    This implies that for every $d \in D$ there is exactly one $k$-ary predicate $P_i$ such that $(d,\dotsc,d) \in P_i^D$.
    Thus, we can write $\exists f.\top_\cDom$ equivalently as ${\exists f,\dotsc,f.P_1}\sqcup{\dotsb}\sqcup{\exists f,\dotsc,f.P_m}$, where each restriction $\exists f,\dots,f.P_i$ repeats $f$ for $k$ times.

	We next give $f_c$ the value $c$ for all elements reachable from a named individual.
    We start by adding all feature assertions $f_c(a,c)$ to~$\Amc'$, for every individual name~$a$ occurring in~\Tmc.
    If \Tmc does not contain any individual names, we instead add only $f_c(a^*,c)$ for a fresh individual name~$a^*$.
    It remains to transfer this value to all reachable elements.
  
	Since \cDom is $\omega$-admissible, equality between two variables $x$, $y$ can be expressed using a formula $\phi_=(x,y)$ that is a disjunction of conjunctions of positive atoms over the signature of $\cDom$ (i.e., not including the additional singleton predicates). 
    Now consider the formula $\phi_=(c,y)$, where $x$ is replaced by the constant~$c$. Since $\phi_=(c,y)$ is equivalent to $c=y$, we can find a single disjunct $\beta(c,y)$ of $\phi_=(c,y)$ such that $\beta(c,y)$ is satisfiable and equivalent to $c=y$; otherwise, $\phi_=(c,y)$ would be satisfied also for values of~$y$ other than~$c$.
    The disjunct $\beta(c,y)$ can be identified in exponential time due to our assumption that \cDom is an \expadm concrete domain with constants.
    
    For every $r\in\NR(\Omc)$, we now obtain the concept~$C_{r,c}$ from $\beta(c,y)$ by replacing $\wedge$ with $\sqcap$ and every atom $P(t_1,\ldots,t_n)$ with $\forall p_1,\ldots,p_n.P$, where $p_i=f_c$ if $t_i=c$ and $p_i=rf_c$ if $t_i=y$, and add the axiom $\top\sqsubseteq C_{r,c}$ to $\Tmc'$. This ensures that, in every model $\Imc$ of $\Amc'$ and $\Tmc'$, for all elements $d$ that are reachable from a named individual by a sequence of role connections, we have that $f_c^\Imc(d)=c$.
	
	We can now replace every concept of the form $\exists f.{=_c}$ in $\Tmc$ with a concept $C_{f=c}$ that is obtained from $\beta(c,y)$ by replacing $\land$ with $\sqcap$ and atoms $P(t_1,\ldots,t_n)$ with $\exists f_1,\ldots,f_n.P$, where $f_i=f_c$ if $t_i=c$ and $f_i=f$ if $t_i=y$.
    Similarly, we can replace feature roles ${=_c}(f)$ by $\gamma_{f=c}$ obtained from $\beta(c,y)$ by replacing $\land$ with $\cap$ and atoms $P(t_1,\ldots,t_n)$ with $P(\alpha_1,\ldots,\alpha_n)$, where $\alpha_i=f_c$ if $t_i=c$ and $\alpha_i=f$ if $t_i=y$.
    The constructed TBox~$\Tmc'$ and ABox~$\Amc'$ are of polynomial size w.r.t.\ the size of~\Tmc since each assertion, concept, or concrete role $f_c(a,c)$, $\exists f_c.\top_\cDom$, $C_{r,c}$, $C_{f=c}$, $\gamma_{f=c}$ is of linear size, the replacements of~$\exists f.{=_c}$ by~$C_{f=c}$ are independent of each other since $\alpha(c,y)$ cannot contain the additional singleton predicates, and similarly for $\gamma_{f=c}$.

    Let now \Imc be a model of~\Tmc.
    By interpreting $f_c$ as the total function with $f_c(d)=c$ for all $d\in\Delta^\Imc$ and, optionally, $a^*$ as an arbitrary element from~$\Delta^\Imc$, we obtain a model of~$\Amc'$ and~$\Tmc'$.
    Conversely, let $\Imc'$ be a model of~$\Amc'$ and~$\Tmc'$.
    We restrict $\Imc'$ to the subdomain of all elements reachable from a named element $a^{\Imc'}$ by a chain of role connections $r^{\Imc'}$ for $r\in\NR$.
    The resulting interpretation~$\Imc''$ is still a model of~$\Amc'$ and~$\Tmc'$ since the evaluation of concepts on $\Delta^{\Imc''}$ does not depend on unconnected elements from $\Delta^{\Imc''}\setminus\Delta^{\Imc'}$ (see \Cref{sec:dls}).
    The new axioms in $\Amc'$ and $\Tmc'$ ensure that $f_c(d)=c$ holds for all $d\in\Delta^{\Imc''}$, and therefore $\Imc''$ is also a model of~\Tmc.
    \qed
\end{proof}

\LemHomogeneous*
\begin{proof}
    Let \Tmc be an \ALCOSCCD TBox and \Amc an ABox containing feature assertions.
    Since predicate assertions can be expressed by TBox axioms, it suffices to show how to simulate the feature assertions by using predicate assertions.
    Let $\Amc'$ result from \Amc by removing all feature assertions and adding the predicate assertions $P(f_1(a_1),\dotsc,f_k(a_k)) \in \Amc'$ for all combinations of feature assertions $f_i(a_i,c_i) \in \Amc$, $i = 1,\dotsc,k$, with $(c_1,\dotsc,c_k) \in P^D$.
    We can check $(c_1,\dotsc,c_k) \in P^D$ in exponential time since \cDom is an \expadm concrete domain with constants.
    However, the size of $\Amc'$ is polynomial in the input, since the predicates of~$\cDom$ are fixed.
  
    It is easy to see that every model of \Tmc and \Amc is also a model of $\Amc'$.
    Conversely, let \Imc be a model of \Tmc and $\Amc'$ and let $\cDom_\Amc$, $\cDom_\Imc$ be the finite substructures of $\cDom$ over the domains
    \begin{center}
      $D_\Amc := \{ c \mid f(a,c) \in \Amc \}$ and $D_\Imc := \{ f^{\Imc}(a) \mid f(a,c) \in \Amc \}$,
    \end{center}
    respectively.
    By definition of~$\Amc'$ and \JEPD, we have $\big(f_1^{\Imc}(a_1^\Imc),\dotsc,f_k^{\Imc}(a_k^\Imc)\big) \in P^D$ iff $(c_1,\dotsc,c_k) \in P^D$, for all combinations of feature assertions $f_i(a_i,c_i)$ in~\Amc.
    By~\JD, this in particular implies that $f_1^\Imc(a_1^\Imc)=f_2^\Imc(a_2^\Imc)$ iff $f_1(a_1,c),f_2(a_2,c)\in \Amc$ for some value $c\in D$, which means that the two substructures have the same number of elements.
    Moreover, by the first equivalence, the mapping $f^{\Imc}(a^\Imc) \mapsto c$ for all $f(a,c) \in \Amc$ is an isomorphism between $\cDom_\Imc$ and $\cDom_\Amc$.
    Since $\cDom$ is homogeneous, there exists an isomorphism $h \colon D \to D$ such that $h(f^{\Imc}(a^\Imc)) = c$ if $f(a,c) \in \Amc$.
    Now, we obtain $\Imc'$ from \Imc by changing the interpretation of feature names to $f^{\Imc'}(d) := h(f^{\Imc}(d))$ iff this value is defined for $f \in \NF$ and $d \in \Delta^{\Imc}$.
    Since $h$ is an isomorphism, we have $C^\Imc=C^{\Imc'}$ for all concepts~$C$, including CD-restrictions and $\suc$-restrictions with feature roles, which shows that $\Imc'$ is a model of~\Tmc.
    Moreover, it also satisfies all feature assertions $f(a,c)\in\Amc$ since $f^{\Imc'}(a^{\Imc'})=h(f^\Imc(a^\Imc))=c$ by construction.
    \qed
\end{proof}

\section{Undecidable Extensions --- Reductions}
\label{app:undecidability}

Throughout this section, we assume that the domain set of \cDom is infinite and that \cDom is \JD (cf.\ \Cref{sec:preliminaries}).
If equality over \cDom is defined by the quantifier-free, equality-free formula $\phi_{=}(x,y)$, we write $\eqterm{f}{\nx g}$ to denote the set term obtained by replacing every atom $P(x_1,\dotsc,x_k)$ in $\phi_{=}(x,y)$ with the feature role $P(\alpha_1,\dotsc,\alpha_k)$, where $\alpha_i = f$ if $x_i = x$ and $\alpha_i = \nx g$ if $x_i = y$ for $i = 1,\dotsc,k$, and every Boolean connective with the corresponding set operation.
\subsection{Local and global cardinality constraints}

We reduce the solvability of a system of Diophantine equations \Emc over the natural numbers to the consistency of a \ALCSCCppD TBox $\Tmc_{\Emc}$.
Without loss of generality, we assume that every equation in \Emc is of the form $x = y \cdot z$, $x = y + z$ or $x = n$ with $x$, $y$, $z$ variables and $n$ a natural number.
The TBox $\Tmc_{\Emc}$ contains a concept name $A_x$ for every variable $x$ occurring in \Emc and a conjunction of CIs $\Tmc_{\efr}$ for every equation $\efr \in \Emc$, built as follows:
\begin{itemize}
    \item if $\efr = (x = n)$, then $\Tmc_{\efr}$ contains the CI $\top \sqsubseteq \sat(|A_x| = n)$;
    \item if $\efr = (x = y + z)$, then $\Tmc_{\efr}$ contains the CI $\top \sqsubseteq \sat(|A_x| = |A_y| + |A_z|)$;
    \item if $\efr = (x = y \cdot z)$, then $\Tmc_{\efr}$ contains the following CIs and \emph{concept definitions} $C \equiv D$, which are a shorthand for $C \sqsubseteq D$ and $D \sqsubseteq C$:
    \begin{itemize}
        \item $A_y \equiv \sat(|r_{\efr}| \ge 1)$ and $A_y \equiv \sat(r_{\efr} = A_z)$ where $r_{\efr}$ is a fresh role name;
        \item $\top \sqsubseteq \sat(s_{\efr} = \eqterm{f_{\efr}}{\nx g_{\efr}})$ and $\top \sqsubseteq \sat(|\eqterm{g_{\efr}}{\nx f_{\efr}}| \le 1)$ as well as $A_x \sqsubseteq \sat(|\eqterm{g_{\efr}}{\nx f_{\efr}}| \ge 1)$ with $s_{\efr} \in \NR$ and $f_{\efr}, g_{\efr} \in \NF$ fresh names;
        \item $\top \sqsubseteq \sat(|r_{\efr}| = |s_{\efr}|)$.
    \end{itemize}
\end{itemize}
The key result for the correctness of the reduction is the following.
\begin{lemma}
    \label{lem:hilbert}
    If $\efr = (x = y \cdot z)$ then $|A_x^{\Imc}| = |A_y^{\Imc}| \cdot |A_z^{\Imc}|$ for every model \Imc of $\Tmc_{\efr}$.
\end{lemma}
\begin{proof}
    First, we show that $r_{\efr}^{\Imc} = A_y^{\Imc} \times A_z^{\Imc}$.
    If $(d,e) \in r_{\efr}^{\Imc}$ holds, then $d \in A_y^{\Imc}$ follows from $A_y \equiv \sat(|r_{\efr}| \ge 1)$; in turn, this implies that $e \in A_z^{\Imc}$ due to $A_y \equiv \sat(r_{\efr} = A_z)$.
    Vice versa, if $d \in A_y^{\Imc}$ and $e \in A_z^{\Imc}$, then $A_y \equiv \sat(r_{\efr} = A_z)$ implies that $(d,e) \in r_{\efr}^{\Imc}$.
    We conclude that $(d,e) \in r_{\efr}^{\Imc}$ iff $d \in A_y^{\Imc}$ and $e \in A_z^{\Imc}$.

    Next, we show that $|s_{\efr}^{\Imc}| = |A_x^{\Imc}|$.
    To show that $|s_{\efr}^{\Imc}| \ge |A_x^{\Imc}|$ holds, we observe that for every $e \in A_x^{\Imc}$ there exists $d \in \Delta^{\Imc}$ such that $f_{\efr}^{\Imc}(d) = g_{\efr}^{\Imc}(e)$ by $A_x \sqsubseteq \sat(|\eqterm{g_{\efr}}{\nx f_{\efr}}| \ge 1)$, and this implies that $(d,e) \in s_{\efr}^{\Imc}$ by $\top \sqsubseteq \sat(s_{\efr} = \eqterm{f_{\efr}}{\nx g_{\efr}})$. Thus, $s_{\efr}^{\Imc}$ contains at least as many tuples as elements of $A_x^{\Imc}$.
    To establish $|s_{\efr}^{\Imc}| \le |A_x^{\Imc}|$, we notice that the function $h$ mapping $(d,e) \in s_{\efr}^{\Imc}$ to $e \in \Delta^{\Imc}$ is an injective function from $s_{\efr}^{\Imc}$ to $A_x^{\Imc}$.
    All the CIs in $\Tmc_{\efr}$ imply that $e \in A_x^{\Imc}$.
    Assuming that $h((d,e)) = h((d',e'))$ and thus $e = e'$, the fact that
    \begin{equation*}
        f_{\efr}^{\Imc}(d) = g_{\efr}^{\Imc}(e) = g_{\efr}^{\Imc}(e') = f_{\efr}^{\Imc}(d')
    \end{equation*}
    together with $\top \sqsubseteq \sat(|\eqterm{g_{\efr}}{\nx f_{\efr}}| \le 1)$ implies $d = d'$, hence that $h$ is injective.
    Finally, we use the CI $\top \sqsubseteq \sat(|r_{\efr}| = |s_{\efr}|)$ to conclude that
    \begin{equation*}
        |A_x^{\Imc}| = |s_{\efr}^{\Imc}| = |r_{\efr}^{\Imc}| = |A_y^{\Imc} \times A_z^{\Imc}| = |A_y^{\Imc}| \cdot |A_z^{\Imc}|,
    \end{equation*}
    where the last identity holds because the domain of \Imc is finite.
    \qed
\end{proof}
\begin{theorem}
    A system of Diophantine equations \Emc has a solution over the natural numbers iff the TBox $\Tmc_{\Emc}$ is consistent.
\end{theorem}
\begin{proof}
    Assume that \Imc is a finite model of $\Tmc_{\Emc}$.
    Then, the assignment $x_\star := |A_x^{\Imc}|$ to every variable $x$ is a solution of all equations $\efr \in \Emc$.
    This is trivial for $\efr = (x = n)$ and $\efr = (x = y + z)$ , and~\Cref{lem:hilbert} shows that this holds for $\efr = (x = y \cdot z)$.

    Vice versa, assume that \Emc has a solution assigning the natural number $x_\star$ to every variable $x$, and that every value assigned by this solution is smaller or equal than the natural number $n_{\Emc}$.
    We define the finite interpretation $\Imc$ with domain $\Delta^{\Imc} := \{1,\dotsc,n_{\Emc}\}$ with $A_x^{\Imc} := \{1,\dotsc,x_\star\}$ for every variable $x$.
    Assuming that $\efr = (x = y \cdot z)$, we define the interpretation of role names
    \begin{equation*}
        r_{\efr}^{\Imc} := A_y^{\Imc} \times A_z^{\Imc} \;\text{and}\; s_{\efr}^{\Imc} := \{ (i, i \cdot j) \mid i \in A_y^{\Imc}, 1 \le j \le z_\star\}.
    \end{equation*}
    To define the interpretation of feature names $f_{\efr}$, $g_{\efr}$, we assume that $h_{\efr}$ is an injective mapping from $A_y^{\Imc}$ to $D$, which always exists since we assumed \cDom to be infinite.
    Then, we define $f_{\efr}^{\Imc}(i) := h_{\efr}(i)$ iff $i \in A_y^{\Imc}$ and $g_{\efr}^{\Imc}(j) := f_{\efr}^{\Imc}(i)$ iff $i \in A_y^{\Imc}$ and $(i,j) \in s_{\efr}^{\Imc}$.
    It is then straightforward to verify that \Imc is a finite model of $\Tmc_{\Emc}$. 
\end{proof}

\begin{corollary}
    If the concrete domain \cDom is infinite and \JD, then the consistency problem for \ALCSCCppD TBoxes is undecidable.
\end{corollary}

\subsection{Transitive roles} We show how to reduce the solvability of a tiling problem $P$ to the consistency of a restricted \SSCC TBox $\Tmc_{P}$.
\begin{definition}
    A \emph{tiling problem} $P := (T,H,V)$ consists of a finite set $T$ of \emph{tile types}
    and binary relations $H, V \subseteq T \times T$ called \emph{horizontal} and \emph{vertical matching} conditions, respectively.
    The function $\pi \colon \mathbb{N} \times \mathbb{N} \to T$ is a 
    \emph{solution} of $P$ if for all $i, j \in \mathbb{N}$ it holds that
    $(\pi(i,j),\pi(i+1,j)) \in H$ and $(\pi(i,j),\pi(i,j+1)) \in V$.
\end{definition}
Adapting the reduction introduced in~\cite{HoSaTo00}, we introduce concept names $A_t$ for $t \in T$ and role names $h$ and $v$ meant to capture the matching conditions of $P$.
If \Imc is a model of $\Tmc_{P}$, we ensure that every $d \in \Delta^{\Imc}$ has exactly one tile type and enforce the existence of exactly one $h$- and one $v$-successor for $d$ using the CI
\begin{equation}
    \begin{split}
        \top &\sqsubseteq {\textstyle \bigsqcup_{t \in T}} A_t \sqcap {\textstyle \bigsqcap_{t \ne t' \in T}} \neg (A_{t} \sqcap A_{t'}) \\
        \top &\sqsubseteq \suc(|h| = 1 \land |v| = 1).
    \end{split}
    \tag{\textsf{successors}}\label{eqn:tiling-successors} 
\end{equation}
The matching conditions of $P$ are encoded in $\Tmc_{P}$ by adding for $t \in T$ the CIs
\begin{equation}
    A_t \sqsubseteq \suc( h \subseteq {\textstyle \bigsqcup_{(t,t') \in H}} A_{t'}) \;\text{and}\; 
    A_t \sqsubseteq \suc(v \subseteq {\textstyle \bigsqcup_{(t,t') \in V}} A_{t'}).
    \tag{\textsf{matching}}\label{eqn:tiling-matching}
\end{equation}
For every model \Imc of $\Tmc_{P}$ we want to ensure that every $v$-successor of a $h$-successor of $d \in \Delta^{\Imc}$ is also a role successor of $d$, and similarly for every $h$-successor of a $v$-successor.
As \Imc must be finitely branching, though, it is not possible to simply use one transitive role $r$ that includes both $h$ and $v$, as this would imply that $d$ has infinitely many successors.
Instead, we introduce role names $h_i$, $v_j$ and $r_{ij}$ with $i,j \in \{0,1\}$, where $r_{ij}$ is a transitive superrole of $h_i$ and $v_j$ thanks to
\begin{equation}
    \mathsf{trans}(r_{ij}) \;\text{and}\; h_i \sqsubseteq r_{ij}\;\text{and}\; v_j \sqsubseteq r_{ij} \;\text{for}\; i,j \in \{0,1\},
    \tag{\textsf{super}}\label{eqn:tiling-super}
\end{equation}
where $r \sqsubseteq s$ is an abbreviation for $\top \sqsubseteq \suc(r \subseteq s)$.
Then, we partition the domain with four concept names $A_{ij}$ with $i,j \in \{0,1\}$ using a similar CI as the one used in~\eqref{eqn:tiling-successors}.
If $d$ is labelled with $A_{ij}$, then none of its successors should be labelled with the same concept $A_{ij}$; further, for $d$ the roles $h_i$ and $v_j$ act as $h$ and $v$ and connect to other individuals $d'$ labelled by concepts $A_{i'j'}$ following the CIs
\begin{equation}
    \begin{split}
        A_{ij} &\sqsubseteq \suc(h \subseteq A_{(1-i)j} \land v \subseteq A_{i(1-j)}) \;\text{and} \\
        A_{ij} &\sqsubseteq \suc(h = h_i \land v = v_j) \;\text{for}\; i,j \in \{0,1\}.
    \end{split}
     \tag{\textsf{local}}\label{eqn:tiling-local}
\end{equation}
These axioms enforce an alternating pattern of roles that ensures that, if $d$ belongs to $A_{ij}$, then it has finitely many $r_{ij}$-successors.

What is left is to ensure for models \Imc of $\Tmc_{P}$ is that the $v$-successor of the $h$-successor of $d \in \Delta^{\Imc}$ is equal to the $h$-successor of the $v$-successor of $d$.
Since both are also successors of $d$ thanks to the presence of transitive roles, we force them to be equal by introducing the CI
\begin{equation}
    \top \sqsubseteq \suc(|h^c \cap v^c| = 1).\tag{$\star$}\label{eqn:tiling}
\end{equation}
\tikzset{help lines/.style={lightgray, very thin}}
\tikzset{ordering/.style={lightgray, very thin}}
\begin{figure}[t]
  \centering
  \begin{tikzpicture}[basic settings]

    \begin{scope}[node distance=1.6cm]
      \node[] (d0) {$A_{00}$};
      \node[] (d1) [above of=d0]{$A_{01}$};
      \node[] (d2) [above of=d1]{$A_{00}$};
      \node[] (d3) [above of=d2]{$\vdots$};
      \node[] (e0) [right of=d0]{$A_{10}$};
      \node[] (e1) [above of=e0]{$A_{11}$};
      \node[] (e2) [above of=e1]{$A_{10}$};
      \node[] (e3) [above of=e2]{$\vdots$};
      \node[] (f0) [right of=e0]{$A_{00}$};
      \node[] (f1) [above of=f0]{$A_{01}$};
      \node[] (f2) [above of=f1]{$A_{00}$};
      \node[] (f3) [above of=f2]{$\vdots$};
      \node[] (g0) [right of=f0]{$\dotsb$};
      \node[] (g1) [above of=g0]{$\dotsb$};
      \node[] (g2) [above of=g1]{$\dotsb$};
    \end{scope}

    \begin{scope}
      \path (d0)
        edge[role=red] node [red,right, midway] {$v_0$} node [gray,align=center,sloped,midway] {\scriptsize $r_{00}$ $r_{10}$} (d1)
        edge[role=blue] node [blue,above,midway] {$h_0$} node [gray,below,midway] {\scriptsize $r_{00}$ $r_{01}$} (e0);
      \path (d1)
        edge[role=red] node [red,right, midway] {$v_1$} node [gray,align=center,sloped,midway] {\scriptsize $r_{01}$ $r_{11}$} (d2)
        edge[role=blue] node [blue,above,midway] {$h_0$} node [gray,below,midway] {\scriptsize $r_{00}$ $r_{01}$} (e1);
      \path (d2)
        edge[role=red] node [red,right, midway] {$v_0$} node [gray,align=center,sloped,midway] {\scriptsize $r_{00}$ $r_{10}$} (d3)
        edge[role=blue] node [blue,above,midway] {$h_0$} node [gray,below,midway] {\scriptsize $r_{00}$ $r_{01}$} (e2);
      \path (e0)
        edge[role=red] node [red,right, midway] {$v_0$} node [gray,align=center,sloped,midway] {\scriptsize $r_{00}$ $r_{10}$} (e1)
        edge[role=blue] node [blue,above,midway] {$h_1$} node [gray,below,midway] {\scriptsize $r_{10}$ $r_{11}$} (f0);
      \path (e1)
        edge[role=red] node [red,right, midway] {$v_1$} node [gray,align=center,sloped,midway] {\scriptsize $r_{01}$ $r_{11}$} (e2)
        edge[role=blue] node [blue,above,midway] {$h_1$} node [gray,below,midway] {\scriptsize $r_{10}$ $r_{11}$} (f1);
      \path (e2)
        edge[role=red] node [red,right, midway] {$v_0$} node [gray,align=center,sloped,midway] {\scriptsize $r_{00}$ $r_{10}$} (e3)
        edge[role=blue] node [blue,above,midway] {$h_1$} node [gray,below,midway] {\scriptsize $r_{10}$ $r_{11}$} (f2);
      \path (f0)
        edge[role=red] node [red,right, midway] {$v_0$} node [gray,align=center,sloped,midway] {\scriptsize $r_{00}$ $r_{10}$} (f1)
        edge[role=blue] node [blue,above,midway] {$h_0$} node [gray,below,midway] {\scriptsize $r_{00}$ $r_{01}$} (g0);
      \path (f1)
        edge[role=red] node [red,right, midway] {$v_1$} node [gray,align=center,sloped,midway] {\scriptsize $r_{01}$ $r_{11}$} (f2)
        edge[role=blue] node [blue,above,midway] {$h_0$} node [gray,below,midway] {\scriptsize $r_{00}$ $r_{01}$} (g1);
      \path (f2)
        edge[role=red] node [red,right, midway] {$v_0$} node [gray,align=center,sloped,midway] {\scriptsize $r_{00}$ $r_{10}$} (f3)
        edge[role=blue] node [blue,above,midway] {$h_0$} node [gray,below,midway] {\scriptsize $r_{00}$ $r_{01}$} (g2);
    \end{scope}

    \begin{scope}
      \path (d0)
        edge[role=purple, densely dotted] node [purple,sloped,below] {$r_{00}$} (e1);
      \path (d1)
        edge[role=purple, densely dotted] node [purple,sloped,below] {$r_{01}$} (e2);
      \path (e0)
        edge[role=purple, densely dotted] node [purple,sloped,below] {$r_{10}$} (f1);
      \path (e1)
        edge[role=purple, densely dotted] node [purple,sloped,below] {$r_{11}$} (f2);
    \end{scope}

    \end{tikzpicture}
    \caption{A representation of the structure enforced using transitive roles.}
    \label{fig:tiling}
\end{figure}
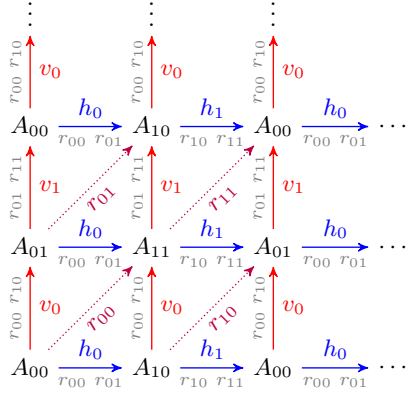
The effect of~\eqref{eqn:tiling-super}, \eqref{eqn:tiling-local} and~\eqref{eqn:tiling} on the models of $\Tmc_{P}$ is showed in~\Cref{fig:tiling}.
We establish the crucial property enjoyed by the models of~$\Tmc_{P}$ below.
\begin{lemma}
    \label{lem:tiling-link}
    If \Imc is a model of $\Tmc_{P}$, then the binary relations $h^{\Imc} \circ v^{\Imc}$ and $v^{\Imc} \circ h^{\Imc}$ coincide and are functional.
\end{lemma}
\begin{proof}
    If \Imc be a model of $\Tmc_{P}$, \eqref{eqn:tiling-successors} guarantees that for every $d \in \Delta^{\Imc}$ that there are four individuals $d_1,d_2,e_1,e_2 \in \Delta^{\Imc}$ such that $(d,d_1) \in h^{\Imc}$, $(d_1,d_2) \in v^{\Imc}$, $(d,e_1) \in v^{\Imc}$ and $(e_1,e_2) \in h^{\Imc}$.
    By~\eqref{eqn:tiling-super} and~\eqref{eqn:tiling-local} we deduce that $d \in A_{ij}^{\Imc}$ iff $(d,d_2) \in r_{ij}^{\Imc}$ and $(d,e_2) \in r_{ij}^{\Imc}$ for $i,j \in \{0,1\}$.
    Moreover, $d_2$ and $e_2$ are both different from $d_1$ and $e_1$, since the concepts $A_{ij}$ are disjoint for $i,j \in \{0,1\}$ and by~\eqref{eqn:tiling-local}.
    Together with~\eqref{eqn:tiling-successors}, this implies that $d_2, e_2 \notin h^{\Imc}(d)$ and $d_2, e_2 \notin v^{\Imc}(d)$.
    Then, we conclude by~\eqref{eqn:tiling} that $d_2 = e_2$ must hold and that $h^{\Imc} \circ v^{\Imc}$ and $v^{\Imc} \circ h^{\Imc}$ coincide.
    \qed
\end{proof}

\begin{lemma}
    \label{lem:tiling-reduction}
    The tiling problem $P$ has a solution iff $\Tmc_{P}$ is consistent.
\end{lemma}
\begin{proof}
    Let \Imc be a model of $\Tmc_{P}$. We define the mapping $\pi \colon \mathbb{N} \times \mathbb{N} \to \Delta^{\Imc}$ inductively, as follows.
    First, let $\pi(0,0)$ be an arbitrary individual in \Imc, which exists since this set must be non-empty.
    Assuming that for $i,j \in \mathbb{N}$ the value $\pi(i,j) := d$ is defined, we define $\pi(i+1,j)$ as the unique $h$-successor of $d$ in \Imc $\pi(i,j+1)$ as the unique $v$-successor of $d$ in \Imc.
    \Cref{lem:tiling-link} guarantees that $\pi$ is well-defined: indeed, the individual $\pi(i+1,j+1)$ is supposed to be both the unique $h$-successor of $\pi(i,j+1)$ and the unique $v$-successor of $\pi(i,j+1)$, and the lemma ensures that these are indeed the same elements.
    Clearly, for all $i, j \in \mathbb{N}$
    \begin{equation}
      \label{eqn:tiling-correct}
      (\pi(i,j),\pi(i+1,j)) \in h^\Imc \; \text{and} \; (\pi(i,j),\pi(i,j+1)) \in v^\Imc.
    \end{equation}
    Using $\pi$, we define $\pi_P \colon \mathbb{N} \times \mathbb{N} \to T$ as $\pi_P(i,j) := t$ iff $\pi(i,j) \in A_t^\Imc$. Then, the fact that $\Imc$ satisfies~\eqref{eqn:tiling-successors} and~\eqref{eqn:tiling-correct} ensures that $\pi_P$ is a solution of $P$.
    
    Next, let $\pi$ be a solution of $P$.
    We define the interpretation $\Imc_\pi$ with domain $\mathbb{N} \times \mathbb{N}$ as follows.
    For each tile type $t \in T$, we set $A_t^{\Imc_\pi}$ as the set of elements $(m,n)$ for which $\pi(m,n) = t$.
    For each element $(m,n)$ in the domain, we add $((m,n),(m+1,n))$ to $h^{\Imc_\pi}$ and $((m,n),(m,n+1))$ to $v^{\Imc_\pi}$.
    Then, writing $(m \equiv i \mod 2)$ to denote that the remainder of the division of $m \in \Nbb$ by $2$ is $i$ (and similarly for $n$ and $j$), we set for $i,j \in \{0,1\}$
    \begin{align*}
        A_{ij}^{\Imc_\pi} &:= \{ (m,n) \in \Delta^{\Imc_\pi} \mid m \equiv i \mod 2,\ n \equiv j \mod 2 \} \\
        h_{i}^{\Imc_\pi} &:= \{ ((m,n),(m+1,n)) \mid (m,n) \in A_{ij}^{\Imc_\pi} \} \\
        v_{j}^{\Imc_\pi} &:= \{ ((m,n),(m,n+1)) \mid (m,n) \in A_{ij}^{\Imc_\pi} \} \\
        r_{ij}^{\Imc_\pi} &:= h_{i}^{\Imc_\pi} \cup v_{j}^{\Imc_\pi} \cup \{ ((m,n),(m + 1,n + 1)) \mid (m,n) \in A_{ij}^{\Imc_\pi}\}
    \end{align*}
    It is straightforward to verify that $\Imc_\pi$ is a model of $\Tmc_{P}$.
    \qed
  \end{proof}

\paragraph{Considerations.} While the above reduction is not syntactically using transitive roles within number restrictions, the semantics of \SSCC is such that every $\suc$-restriction implicitly ranges over the transitive roles, as well.
If we disallow role complement as used in~\eqref{eqn:tiling}, we could still enforce correctness of the reduction by adding
$A_{ij} \sqsubseteq \suc(|A_{(1-i)(1-j)}| = 1)$
to $\Tmc_P$, which works because of the implicit ranging over role successors induced by the semantics of \SSCC.

One may ask if lifting this semantic condition allows us to regain decidability.
If we consider the extension $\SSCC^{++}$ of \SSCC, defined in the same spirit of \ALCSCCpp w.r.t. \ALCSCC, we can replace~\eqref{eqn:tiling} with $A_{ij} \sqsubseteq \sat(|r_{ij} \cap A_{(1-i)(1-j)}| = 1)$ to cause undecidability.
While only using coefficients $0$ or $1$, this variant requires an explicit usage of transitive roles in number restrictions.
It is unclear if the same can effect can be achieved while disallowing transitive roles within numerical constraints.

\end{document}